\documentclass[11pt]{article}

\newif\ifblind
\blindfalse

\newif\ifdraft
\draftfalse

\usepackage{rpmacros}
\usepackage[dvipsnames,svgnames,x11names,hyperref]{xcolor}
\usepackage{stmaryrd}
\usepackage{amsthm}
\usepackage{xfrac}
\usepackage{mathpazo}
\usepackage{gitinfo}
\usepackage{enumitem}
\usepackage{fullpage}
\usepackage{thmtools}
\usepackage[T1]{fontenc} 
\usepackage[colorlinks]{hyperref}
\hypersetup{
  linkcolor=[rgb]{0,0,0.4},
  citecolor=[rgb]{0, 0.4, 0},
  urlcolor=[rgb]{0.6, 0, 0}
}
\usepackage{lipsum}
\usepackage{calrsfs,bbm}
\usepackage{multirow}
\usepackage{setspace}

\usepackage{microtype}

\usepackage[linesnumbered,vlined,ruled]{algorithm2e}
\usepackage{mathtools}

\usepackage{amsthm}
\usepackage{thmtools,thm-restate}
\usepackage[nameinlink,capitalise]{cleveref}

\numberwithin{equation}{section}
\declaretheoremstyle[bodyfont=\it,qed=\qedsymbol]{noproofstyle}

\declaretheorem[numberlike=equation]{observation}

\declaretheorem[name=Observation,numbered=no]{observation*}

\declaretheorem[numberlike=equation]{theorem}

\declaretheorem[name=Theorem,numbered=no]{theorem*}

\declaretheorem[numberlike=equation]{lemma}
\declaretheorem[name=Lemma,numbered=no]{lemma*}

\declaretheorem[numberlike=equation]{corollary}
\declaretheorem[name=Corollary,numbered=no]{corollary*}

\declaretheorem[name=Proposition,numbered=no]{proposition*}

\declaretheorem[name=Claim,numbered=no]{claim*}

\declaretheorem[name=Conjecture,numbered=no]{conjecture*}

\declaretheorem[name=Question,numbered=no]{question*}

\declaretheoremstyle[bodyfont=\it,qed=$\lozenge$]{defstyle} 

\declaretheorem[numberlike=equation,style=defstyle]{definition}
\declaretheorem[unnumbered,name=Definition,style=defstyle]{definition*}

\declaretheorem[unnumbered,name=Example,style=defstyle]{example*}

\declaretheorem[unnumbered,name=Notation=defstyle]{notation*}

\declaretheorem[unnumbered,name=Construction,style=defstyle]{construction*}

\declaretheorem[unnumbered,name=Remark,style=defstyle]{remark*}

\renewcommand{\phi}{\varphi}
\renewcommand{\epsilon}{\varepsilon}

\newcommand{\ord}{\operatorname{ord}}

\newcommand{\SP}{\Sigma\Pi}
\newcommand{\SPsize}[1]{(\SP)^k\operatorname{-size}}

\usepackage{nth}
\usepackage{intcalc}
\usepackage{etoolbox}
\usepackage{xstring}

\usepackage{ifpdf}
\ifpdf
\else
\usepackage[quadpoints=false]{hypdvips}
\fi

\newcommand{\ECCCurl}[2]{http://eccc.hpi-web.de/report/\ifnumcomp{#1}{>}{93}{19}{20}#1/#2/}

\newcommand{\shortECCC}[2]{\texttt{\href{http://eccc.hpi-web.de/report/\ifnumcomp{#1}{>}{93}{19}{20}#1/#2/}{eccc:TR#1-#2}}}

\newcommand{\parseECCC}[1]{
\StrSubstitute{#1}{TR}{}[\tmpstring]%
\IfSubStr{\tmpstring}{/}{ 
\StrBefore{\tmpstring}{/}[\ecccyear]%
\StrBehind{\tmpstring}{/}[\ecccreport]%
}{
\StrBefore{\tmpstring}{-}[\ecccyear]%
\StrBehind{\tmpstring}{-}[\ecccreport]%
}%
\shortECCC{\ecccyear}{\ecccreport}}

\newcounter{todo}

\ifdraft
\newcommand{\RPnote}[1]{\refstepcounter{todo}\textcolor{WildStrawberry}{\guillemotleft RP: #1 \guillemotright\addcontentsline{tod}{subsection}{[RP]~#1}}}
\newcommand{\SBnote}[1]{\refstepcounter{todo}\textcolor{BlueGreen}{\guillemotleft Shubham: #1 \guillemotright\addcontentsline{tod}{subsection}{[SB]~#1}}}
\newcommand{\ASnote}[1]{\refstepcounter{todo}\textcolor{Brown}{\guillemotleft Amit: #1 \guillemotright\addcontentsline{tod}{subsection}{[AS]~#1}}}
\newcommand{\gitinfonotecolour}{Gray}
\newcommand{\easteregg}{}
\else
\newcommand{\RPnote}[1]{}
\newcommand{\SBnote}[1]{}
\newcommand{\ASnote}[1]{}
\newcommand{\gitinfonotecolour}{white}
\newcommand{\easteregg}{"The product of mathematics is clarity and understanding. Not theorems by themselves." -- Bill Thurston}
\fi

\newcommand{\ignore}[1]{}
\newcommand{\gitinfonote}{git info:~\gitAbbrevHash\;,\;(\gitAuthorIsoDate)\; \;\gitVtag}

\makeatletter

\newcommand\listtodoname{List of todos}
\newcommand\listoftodos{%
  \section*{\listtodoname}\@starttoc{tod}}
\makeatother

\allowdisplaybreaks
\newcommand{\stewardnote}{%
  The origination of the ideas was via AI tools. We list ourselves as \emph{stewards} to emphasise this, but take responsibility in providing an exposition of the external proof with the clarity and understanding that we believe would help any member of the community. }

\title{Hitting sets From Derivations:\\Waring Rank and Noncommutative ABPs}

\makeatletter
\newcount\pact@n
\newcount\pact@cols
\newcommand{\steward}[1]{\def\pact@name{#1}\pact@opts}
\newcommand{\pact@opts}{\@ifnextchar[\pact@inst{\pact@add{}{}}}
\def\pact@inst[#1]{\@ifnextchar[{\pact@mail{#1}}{\pact@add{#1}{}}}
\def\pact@mail#1[#2]{\pact@add{#1}{#2}}
\newcommand{\pact@add}[2]{\expandafter\pact@store\expandafter{\pact@name}{#1}{#2}}
\newcommand{\pact@store}[3]{%
  \global\advance\pact@n by 1
  \expandafter\gdef\csname pact@c\the\pact@n\endcsname{%
    \begin{tabular}[t]{@{}c@{}}#1
      \if\relax\detokenize{#2}\relax\else\\ {\small #2}\fi
      \if\relax\detokenize{#3}\relax\else\\ {\small\texttt{#3}}\fi
    \end{tabular}}}
\newcommand{\stewardlist}{%
  \ifnum\pact@n<4 \pact@cols=\pact@n
  \else\pact@cols=2 \fi
  \gdef\pact@body{}%
  \@tempcnta=1\relax
  \@whilenum\@tempcnta<\pact@n\do{%
    \expandafter\g@addto@macro\expandafter\pact@body
      \expandafter{\csname pact@c\the\@tempcnta\endcsname}%
    \@tempcntb=\@tempcnta \divide\@tempcntb by \pact@cols
    \multiply\@tempcntb by \pact@cols
    \ifnum\@tempcntb=\@tempcnta
      \g@addto@macro\pact@body{\\[2.5ex]}%
    \else
      \g@addto@macro\pact@body{&}%
    \fi
    \advance\@tempcnta by 1\relax}%
  \expandafter\g@addto@macro\expandafter\pact@body
    \expandafter{\csname pact@c\the\pact@n\endcsname}%
  \begin{tabular}{@{}*{\pact@cols}{c@{\hspace{2em}}}@{}}
    \pact@body
  \end{tabular}}
\makeatother

\steward{Shubham Bhardwaj\thanks{Tata Institute of Fundamental Research, Mumbai, India. Email: \texttt{\{shubham.bhardwaj,ramprasad\}@tifr.res.in}.  Research supported by the Department of Atomic Energy, Government of India, under project number RTI400112, and Google and SERB Research Grants.  }}
\steward{Foram Lakhani\thanks{IIT Kanpur, India. Email: \texttt{foramp23@iitk.ac.in}.}}
\steward{Partha Mukhopadhyay\thanks{Chennai Mathematical Institute, India. Email: \texttt{\{partham,amitks\}@cmi.ac.in}}}
\steward{Ramprasad Saptharishi\footnotemark[2]}
\steward{Pratik Shastri\thanks{IT University of Copenhagen, Denmark. Email: \texttt{prasha@itu.dk}.}}
\steward{Amit Kumar Sinhababu\footnotemark[4]}
\author{%
  {\large \textsc{Stewards}\thanks{\stewardnote}}
  \\[1ex]
  \stewardlist
}
\date{}

\begin{document}

\maketitle

\begin{abstract}
    We give an explicit hitting set of size $\poly(n,d,r)$ for the class of $n$-variate degree-$d$ polynomials whose partial derivative space is bounded by $r$, over any field $\F$ of characteristic zero. In particular, this yields a polynomial sized hitting set for the class of depth-$3$ powering circuits. 

    The proof extends to the setting of noncommutative ABPs and provides a polynomial sized hitting set (of matrices) for the class of noncommututive polynomials that can be computed by polynomial sized noncommutative ABPs (this was also one of the results reported by Open AI~\cite{OpenAI26} in their release on 6 Oct 2026). 

    The main technical insight is the construction of a ``formal derivation'' and properties of the associated Wronskian with respect to this derivation, which was previously studied by Moura \cite{Moura_2004} in a very different context. The proofs in this paper are elementary and completely self-contained.

    \medskip 

    \textbf{AI disclosure:} A subset of the stewards obtained a weaker statement (hitting sets for polynomials with bounded partial derivative space) during conversations \cite{astra_proof} with OpenAI GPT-6 Astra. Extensions of this idea to the broader class of noncommutative ABPs were obtained by the other stewards (via further conversations with Astra), prior to OpenAI's dump on 6 Oct 2026.

    The proofs presented in this writeup is a rewriting (in the authors' words) of the proofs for both settings obtained by the models in a form that we believe is understandable to researchers. Although the proof is technically identical to OpenAI's release, we believe this exposition is self-contained, simpler, and more valuable to the community. The absurdity of the situation is not lost on us.
\end{abstract}

\section{Introduction}

Algebraic circuits are the most natural computational model for studying the complexity of multivariate polynomials via the number of basic operations (additions, multiplications) required to compute it. The field of algebraic complexity broadly deals with the classification of polynomials based on their computational complexity, and also on algorithmic questions regarding polynomials provided via such computations. 

Two of the most important aspects of this study are the question of lower bounds (i.e., finding explicit polynomials that cannot be computed by small algebraic circuits), and the algorithmic task of polynomial identity testing (i.e., given a circuit $C$, algorithmically check if $C$ is computing the zero polynomial). In the context of polynomial identity tests, there are two types of algorithms --- whitebox PITs (where the algorithm is given the entire circuit description), or blackbox PITs (where we are only provided evaluation access to the circuit). Although the questions of proving lower bounds and constructing deterministic PITs appear to be of very different flavours, they are intimately connected to each other. Several classical results \cite{HS80,A05a,KI04,ChouKS19,GKSS19} show that a strong enough lower bound for circuits can be used to obtain efficient deterministic polynomial identity tests, and vice-versa. Although similar statements when specialised to restricted circuit classes are known in some contexts \cite{DSY09,Oliveira16,ChouKS19,BKRRSS26}, this has broadly been a guiding principle for the study of lower bound and PITs for subclasses of circuits. Historically, a polynomial identity test for a class $\mathcal{C}$ has been a successor to a lower bound for the class $\mathcal{C}$, and is often heavily inspired by the techniques used in the lower bound proof. Paradoxically, even `natural' ``barriers'' for proving algebraic circuit lower bounds stem from the conjectured existence of \emph{succinct polynomial identity tests} \cite{GKSS17,FSV18}. 

Over the recent years, there has been quite a lot of progress made towards proving lower bounds for natural subclasses of algebraic circuits, and also constructing whitebox and blackbox PITs for them \cite{AGKS15,GKST15,GG20,LST21,AF22,GOSSS26}. One subclass of circuits that has been just shy of a full resolution has been the class of \emph{depth-$3$ powering circuits}. 

\paragraph{Depth-$3$ powering circuits:} A depth-$3$ powering circuit computes a polynomial of the form $f = c_1 \ell_1^d + \cdots + c_s \ell_s^d$ where each $\ell_i$ is a linear polynomial over the underlying variables. This is the circuit analogue of the classical notion of Waring rank, where the Waring rank of a polynomial $f$ is the smallest $s$ such that $f$ can be written as sum of $s$ powers of linear forms. Saxena~\cite{S08b} gave the first lower bound for this class of circuits, and also a polynomial time white-box PIT. A different white-box PIT was also given by Kayal~\cite{K10}. Agrawal, Saha and Saxena~\cite{ASS13}, and Forbes and Shpilka~\cite{FS13} gave the first $n^{O(\log n)}$-size hitting set for this class (among constructions of quasi-polynomial sized hitting sets for other more general classes), and Forbes, Saptharishi and Shpilka~\cite{FSS14} improved this to $n^{O(\log\log n)}$.

The key weakness of the above model that enabled all the lower bounds and PITs was the observation that any $k$-th order partial derivative of $f = c_1 \ell_1^d + \cdots + c_s \ell_s^d$ is a linear combination of $\ell_1^{d-k},\ldots, \ell_s^{d-k}$ and thus the space of all possible partial derivatives of $f$ (which we shall denote by $\partial^{=*}(f)$) is only $O(sd)$-dimensional. The hitting sets constructed by \cite{ASS13,FSS14} also extend to the (potentially) more general class of polynomials $f$ whose space of all possible partial derivatives is low-dimensional. It continues to remain open if such polynomials with low-dimensional space of partial derivatives admit a not-too-large representation as a depth-$3$ powering circuit. 

The first result of this paper is an explicit polynomial-sized hitting set for this class of polynomials over a characteristic zero field. 

\begin{theorem}
\label{thm:main-thm}
Let $\F$ be a field of characteristic zero. For any $n,r \geq 0$, there is an explicit polynomial map $(G_1(t),\ldots, G_n(t)) \in \F[t]^n$ with $\deg G_i \leq O(nr^2)$ such that if $f(x_1,\ldots, x_n)$ is a polynomial such that $\dim \partial^{=*} f \leq r$, then 
\[
f(x_1,\ldots, x_n) = 0 \Leftrightarrow f(G_1(t),\ldots, G_n(t)) = 0.
\]
\end{theorem}

In particular, the above yields a $\poly(n,d,s)$-sized hitting set for the class of size $s$, degree $d$ depth-$3$ powering circuits. In fact, $G_i(t)$ can just be instantiated to be the \emph{truncated Taylor series} of $\ln(1 + it)$:
\[
G_i(t) = it - \frac{(it)^2}{2} + \frac{(it)^3}{3} - \cdots + (-1)^{B-1}\frac{(it)^B}{B} \quad\text{for $i = 1,\ldots, n$}
\]
where $B = nr^2 + r$. \\

\paragraph{Using power series for constructing hitting-set generators: }
The general structure of constructing a hitting-set generator via power series is reminiscent of the randomised PIT algorithms of Chen and Kao \cite{CK97}, or Lewin and Vadhan~\cite{LV98}. We describe the setting of \cite{CK97} (the setting of \cite{LV98} is very similar with irreducible polynomial playing the role of primes). They show that if $f(x_1,\ldots, x_n)$ is a nonzero polynomial, then 
\[
f(\pm\sqrt{p_{1,1}} \pm \cdots \pm\sqrt{p_{1,\ell}}, \cdots, \pm\sqrt{p_{n,1}} \pm \cdots \pm\sqrt{p_{n,\ell}} ) \neq 0
\]
where $p_{i,j}$'s are distinct primes with $\ell = O(\log \deg(f))$. However, an efficient algorithm cannot evaluate $f$ at these irrationals of infinite precision. \cite{CK97} show that if the signs are randomly chosen, we can achieve an error of $\epsilon$ as long as we truncate the above irrationals to precision $B(\epsilon)$. This results in randomised algorithm for PIT where reducing the error is achieved by increasing the precision (and thereby the running time) but keeping the number of random bits fixed. 

The PITs in this paper is similar in the sense that the generator is appropriate truncations of infinite objects such as algebraically independent power series (e.g. $\inparen{ P_i(t) := \log(1 + it) \;:\; i\in [n]}$) or linear operators in the noncommutative setting. For the class of polynomials with low-dimensional partial derivative space, we show that 
\[
0 \neq f \implies f(P_1(t), \ldots, P_n(t)) \neq 0 \bmod t^B
\]
for a not-too-large $B$. For the noncommutative setting, the linear operators are closely related to the above, but are given by certain \emph{iterated integrals}. 

Mrinal Kumar pointed out to us that $f(e^{zt}, e^{z^2 t}, \ldots, e^{z^n t}) \neq 0 \bmod t^s$ for any nonzero $s$-sparse polynomial (thus giving a bivariate, low degree, hitting set generator for sparse polynomials). It would be interesting to see if this applies to a more general class of polynomials.

\paragraph{Noncommutative algebraic branching programs:} Next, we will look at the PIT question for noncommutative models of computation. In particular, we look at noncommutative ABPs. Fix a field $\mathbb{F}$ and a set of variables $X = \{x_1, \ldots, x_n\}$. A noncommutative polynomial over $\mathbb{F}$ in $X$ is a finite $\mathbb{F}$-linear combination of \emph{words} in $X^*$. The set of all noncommutative polynomials in $X$ forms a ring under addition and noncommutative products. It is denoted by $\mathbb{F}\langle X\rangle$.

Nisan, in his seminal paper~\cite{conf/stoc/Nisan91}, introduced the model of noncommutative Algebraic Branching Programs (also called noncommutative ABPs)\footnote{To the best of our knowledge, this was the first time algebraic branching programs were explicitly introduced as a model of computation, noncommutative or otherwise.}. This is a model for computing noncommutative polynomials. A noncommutative ABP is a layered directed acyclic graph with two distinguished vertices, a source and a sink. The first layer contains only the source and the last only the sink. Every edge goes from one layer to the next and is labeled by a homogeneous linear polynomial. The weight of a source to sink path is the (order-respecting) product of the edge labels on it, and the ABP is said to compute $f$. Nisan proved exponential lower bounds for this model~\cite{conf/stoc/Nisan91}. The PIT problem for noncommutative ABPs has also been studied extensively. More than a decade after Nisan's work, Raz and Shpilka gave a deterministic polynomial-time whitebox PIT algorithm~\cite{RS05} for this problem. In the blackbox setting, a hitting set for noncommutative ABPs cannot have commuting entries, so one allows evaluations at matrices of polynomial dimension~\cite{FS13}. Forbes and Shpilka~\cite{FS13} studied the nc-ABP PIT problem in the black-box setting and gave a quasipolynomial size matrix hitting set. 

There has also been a long line of work extending the study of noncommutative computation to models allowing inversion gates, beginning with the work of Hrube\v{s} and Wigderson \cite{HW15}. In a breakthrough, a deterministic polynomial-time algorithm for noncommutative rational identity testing was obtained by Garg, Gurvits, Oliveira, and Wigderson \cite{GGOW16, GGOW20} using operator scaling. It was followed by algorithms based on purely algebraic techniques \cite{IQS18}, and using convex optimization \cite{HH21}. In the black-box setting, randomized polynomial-time algorithm follows from the work of Derksen and Makam~\cite{DM17}, while deterministic quasipolynomial-time hitting sets for noncommutative rational formulas were obtained in \cite{ACM24}. Despite considerable effort, the Forbes-Shpilka algorithm has remained unbeaten for nc-ABP PIT, in the blackbox setting. The second result of this paper is an explicit polynomial size hitting set for the noncommutative ABPs.

\begin{theorem}\label{thm:appendix-main}
Let $\F$ be a field of characteristic zero.
There is a polynomial-time algorithm which, when given inputs $n,s\in \mathbb{N}$ in unary, outputs a tuple $T_1,\cdots,T_n$ of square matrices with rational entries such that for every nonzero noncommutative polynomial $f\in\F\langle x_1, \ldots, x_n\rangle$ computable by a homogeneous noncommutative ABP of size at most $s$, 
$f(T_1, \ldots, T_n)\neq 0$. The $T_i$'s have dimension at most $ns^2 + s$.
\end{theorem}

Extending the above result to the non-homogeneous case is straightforward using a standard argument. We provide brief details for completeness.

\begin{corollary}\label{cor:appendix-main}
There is a polynomial-time algorithm which, given
$n,s\in\mathbb{N}$ in unary, outputs an explicit hitting
set $\mathcal{H}$ of at most $s$ tuples of square matrices
over $\F$, each of dimension $O(ns^6)$, such that
for every nonzero polynomial
$f\in\F\langle x_1,\ldots,x_n\rangle$
computable by a noncommutative ABP of size at most $s$, there exists
$(A_1,\ldots,A_n)\in\mathcal{H}$ satisfying
\[
    f(A_1,\ldots,A_n)\neq 0.
\]
\end{corollary}

\begin{proof}
Let $f=f_0+f_1+\cdots+f_d$, where $f_k$ denotes the
homogeneous component of degree $k$.
Since $f$ is computed by an ABP of size at most $s$, we have $d\leq s$. By the standard homogenization construction, each
$f_k$ can be computed by a homogeneous ncABP of size
at most $O(sd^2)=O(s^3)$. Let $S=O(s^3)$ fixed suitably. Applying \cref{thm:appendix-main} with parameters $n,S$, we obtain
matrices
\[
    T_1,\ldots,T_n\in\F^{B\times B},
    \qquad
    B=nS^2+S=O(ns^6),
\]
such that $f_k(T_1,\ldots,T_n)\neq 0$ whenever
$f_k\neq 0$.

Now introduce a commuting indeterminate $z$ and
consider the matrix polynomial
\[
\begin{aligned}
    F(z)
    &:=f(zT_1,\ldots,zT_n)\\
    &=\sum_{k=0}^{d}z^k
      f_k(T_1,\ldots,T_n).
\end{aligned}
\]
By \cref{thm:appendix-main}, its evaluation
$f_k(T_1,\ldots,T_n)$ is nonzero.
Thus $F(z)$ is a nonzero matrix polynomial of
degree at most $d\leq s$. Consequently, there exists
$z\in\{0,1,\ldots,s\}$ such that $F(z)\neq 0$.
Therefore, the set
\[
    \mathcal{H}
    :=
    \left\{
       (zT_1,\ldots,zT_n):
       z\in\{0,1,\ldots,s\}
    \right\}
\]
is a hitting set for all nonzero noncommutative ABPs of size at most $s$.
\end{proof}

\subsection{Proof overview}

It would be instructive to begin with the setting of polynomials with low-dimensional partial derivatives as it illustrates the core ideas, which also translates to the noncommutative ABP setting.  

Let $f(x_1,\ldots, x_n)$ be some nonzero polynomial. We would like to understand if $f(P_1(t), \ldots, P_n(t))$ can be divisible by $t^B$ for a large $B$ (for if not, we can just truncate each $P_i(t)$ to $B$-terms, and we have a hitting set generator). Let us think of this substitution as an $\F(t)$-linear map (since this perspective would be useful for the noncommutative setting):
\begin{align*}
\Phi \colon \F(t)[x_1,\ldots, x_n] &\rightarrow \F\indparen{t}\\
\Phi(g(t, x_1,\ldots, x_n)) & = g(t, P_1(t), \ldots, P_n(t)). 
\end{align*}
By the chain rule of derivatives, 
\begin{align*}
\frac{d}{dt}\Phi(g) & = \frac{d}{dt}g(t, P_1,\ldots, P_n) = \Phi\inparen{\partial_t g + \sum_{i=1}^n \frac{d P_i}{dt} \cdot \partial_{x_i}g}.
\end{align*}
From the above, let us define a \emph{derivation} $D: \F(t)[x_1,\ldots, x_n] \rightarrow \F\indparen{t}[x_1,\ldots, x_n]$ given by 
\[
D = \partial_t + \sum_{i=1}^n b_i(t) \cdot \partial_{x_i}
\]
where $b_i(t) = \frac{d P_i}{dt}$ so that, for any $g \in \F(t)[x_1,\ldots, x_n]$, we have $\frac{d}{dt}\Phi(g) = \Phi(Dg)$. One particular derivation that would be extremely useful in this setting is the following:
\[
D = \partial_t + \sum_{i=1}^n \frac{i}{1 + it} \cdot \partial_{x_i}.
\]
Derivations of similar shape was studied by \cite{Moura_2004}, and the main reason we would be using this is for the following miraculous property --- if $h(t, x_1,\ldots, x_n) \in \F(t)[x_1,\ldots, x_n]$ with $Dh = 0$, then $h$ is actually a scalar in $\F$. This property is sometimes referred to as $D$ having \emph{a trivial ring of constants}. For $b_i(t) = \tfrac{i}{1 + it}$, the corresponding $P_i(t) = \int b_i(t) dt = \ln(1 + it)$. The above property that $Dh = 0 \implies h \in \F$ can in fact be used to show that $\setdef{\ln(1 + it)}{i\in [n]}$ are algebraically independent power series in $t$. Therefore, for any nonzero polynomial $f(x_1,\ldots, x_n)$, we have $f(P_1,\ldots, P_n)$ is a nonzero element of $\F\indsquare{t}$. We now want a way of arguing that if $f$ was \emph{simple}, then $f(P_1(t),\ldots, P_n(t)) \neq 0 \bmod{t^B}$ for a moderate $B$. 

Now suppose the $f$ has a low-dimensional partial derivative space, and say $\set{f_1,\ldots, f_r}$ is a basis of this space (with $f_r = f$). Since $f_1,\ldots, f_r$ are linearly independent polynomials, we also have that $\Phi(f_1), \ldots, \Phi(f_r)$ are linearly independent elements of $\F\indsquare{t}$. Therefore, the standard Wronskian $W(\Phi(f_1), \ldots, \Phi(f_r))$ is nonzero power series in $t$ (although this is folklore, Bostan and Dumas~\cite{BD10} provide a clean proof of this statement). This Wronskian matrix (for which $W(\Phi(f_1), \ldots, \Phi(f_r))$ is the determinant of) is just $\Phi$ applied point-wise to the following $D$-Wronskian matrix $\mathcal{W}_D(f_1,\ldots, f_r)$:

\[
\mathcal{W}_D(f_1,\ldots, f_r) = \begin{bmatrix}
f_1 & f_2 & \cdots & f_r\\
D(f_1) & D(f_2) & \cdots & D(f_r)\\
\vdots & \vdots & \ddots & \vdots \\
D^{r-1}(f_1) & D^{r-1}(f_2) & \cdots & D^{r-1}(f_r)
\end{bmatrix}.
\]
Since the set $\set{f_1,\ldots, f_r}$ is closed under taking derivatives, one can show that the above matrix factorises as
\[
\mathcal{W}_D(f_1,\ldots, f_r) = R(t) \cdot F
\]
where $R(t)$ is an $r\times r$ matrix consisting entries that are rational functions in $t$ (of moderate numerator and denominator degrees), and $F$ is an upper-triangular matrix with $1$'s on the diagonal, with off-diagonal entries involving polynomials in $x$. A corollary of this structure is that $W_D(f_1,\ldots, f_r) = \det \mathcal{W}_D(f_1,\ldots, f_r)$ is just a rational function in $t$ and does not depend on $x$ at all! Hence, $W_D(f_1,\ldots, f_r) = W(\Phi(f_1), \ldots, \Phi(f_r)) = \frac{A(t)}{B(t)}$ where the degrees of $A, B$ are not-too-large. If $\Phi(f) = 0 \bmod t^B$, every entry in the last column of $\mathcal{W}((\Phi(f_1), \ldots, \Phi(f_r)))$ is divisible by $t^{B-r+1}$ and this would force $t^{B-r+1} \mid A(t)$. If $B$ is suitably large, this would be impossible. 

\paragraph{The noncommutative ABP setting: }

Turns out, the same proof works almost exactly along the same lines for the noncommutative ABP setting. Here, we need to define $\partial_{x_i}$ as the ``left-derivative'' (i.e., for any word $w$ we have $\partial_{x_i}(x_j w) = w$ if $j=i$, and zero otherwise). The same derivation operator defined above also satisfies the following --- if $h(t, x_1,\ldots, x_n) \in \F(t)\inangle{x_1,\ldots, x_n}$ is a noncommutative polynomial with $Dh = 0$, then $h \in \F$. That is, the derivation has a trivial ring of constants in the noncommutative setting as well. 

We now wish to define an $\F(t)$-linear map $\Psi: \F(t)\inangle{x_1,\ldots, x_n} \rightarrow \F\indparen{t}$ with the same property as earlier that $\frac{d}{dt} \Psi(f) = \Psi(Df)$. Analogous to earlier, we would be able to show that $D$ having a trivial ring of constants would translate to $\Psi$ being injective. In other words, $\Psi$ will map every noncommutative polynomial to a nonzero power series. 

Let us normalise $\Psi$ to ensure $\Psi(1) = 1$, and try to understand what power series $\Psi$ should map a word $w$ to. Certainly, for each $i \in [n]$, we need to ensure 
\begin{align*}
\frac{d}{dt}\Psi(x_i) & = \Psi(Dx_i) = b_i(t)\Psi(1) = b_i(t)\\
\implies \Psi(x_i) & = \int b_i(t) dt = P_i(t)
\end{align*}
like earlier. Now, to understand what $\Psi$ should map words of larger length to, consider $x_i w$ for some word $w$. If we want  $\tfrac{d}{dt}\Psi(f) = \Psi(Df)$ for all $f$, then
\begin{align*}
\frac{d}{dt} \Psi(x_i w) & = \Psi(D (x_i w)) = b_i(t) \Psi(w)\\
\implies \Psi(x_i w) & = \int b_i(t) \Psi(w) dt
\end{align*}
which give the \emph{iterated integrals} that we alluded to. In fact, the above can be thought of as replacing $x_i$ by the linear operator $T_i$ on the infinite dimensional vector space $\F\indsquare{t}$ that acts as 
\begin{quote}
$T_i(g(t)) = $ ``multiply $g(t)$ with $b_i(t)$ and integrate formally''. 
\end{quote}
Now that we defined the injective map $\Psi\colon\F(t)\inangle{x_1,\ldots, x_n}$ to $\F\indparen{t}$, let us focus on the polynomial for which we want a PIT. If $f$ is computable by a size $s$ noncommutative branching program $C$, then the partial derivative space (under left-derivatives) of $f$ is in fact low-dimensional --- a spanning set for this space is just $\set{f_1,\ldots, f_s}$ where $f_i$ is the noncommutative polynomial computed by $C$ between the $i$-th vertex and the sink. Once again, we wish to understand the Wronskian $W(\Psi(f_1), \cdots, \Psi(f_s))$, the associated matrix $\mathcal{W}(\Psi(f_1), \ldots, \Psi(f_s))$, and the related $D$-Wronskian $\mathcal{W}_D(f_1,\ldots, f_s)$. Along very similar lines, turns out we can factorise the matrix $\mathcal{W}_D(f_1,\ldots, f_s)$ as
\[
\mathcal{W}_D(f_1,\ldots, f_s) = R(t) \cdot F
\]
where $R(t)$ is once again a matrix of rational functions in $t$ of moderate numerator and denominator degree, and $F$ is an upper-triangular matrix with $1$'s on the diagonal and off-diagonal entries being elements of $\F\inangle{x_1,\ldots, x_n}$. Hence, we get that $W(\Psi(f_1), \ldots, \Psi(f_s)) = \det R(t) = \frac{A(t)}{B(t)}$ of not-too-large degrees. Just like earlier, if $\Psi(f) = 0 \bmod t^B$, then we would be forced to admit that $t^{B-s+1} \mid A(t)$ and this is impossible if $B$ is chosen large enough. 

Therefore, we can truncate $\Psi$ to just $B$ terms of the power series, and this results in the linear operators $T_i$'s being truncated to $B\times B$ matrices giving us the hitting set. 

\subsection{Classical sources of the key ideas}

As is often the case with AI-generated proofs, the key ideas used in this paper have been discovered multiple times in completely different contexts. We try to give our account of them to the best of our understanding. 

The notion of formal derivations opens into a general study of \emph{derivation fields} and the specific derivation used is an example of something called Shamsuddin's derivations (see \cite{nowicki} for more details). The analysis of the order of vanishing via Wronskians can be traced to Moura~\cite{Moura03, Moura_2004}. The iterated integrals used in the noncommutative ABP setting are also understood as \emph{hyperlogarithms}, or the Chen-Fliess series \cite{Chen1954, Fliess1981} and their linear independence appears to have been studied extensively \cite{Panzer2015, DDMS11}.
The article by Espinosa et al~\cite{espinosa2026chen} provide some historical roots of Chen-Fliess series.

In this paper, we take the approach of giving a self-contained proof of the hitting set constructions in the style of a motivated exposition rather than merely citing the above papers, but we wish to stress that all of these ideas have classical roots in fields perhaps slightly removed from algebraic complexity. We give a version of the proof in this classical context in the \cref{sec:classical-proof} for readers who might be more familiar with the associated backgrounds. 

\subsection*{Notation}

\begin{itemize}\itemsep0pt
    \item Throughout this paper we assume that we will be working with a field of characteristic zero. We use $\F\indsquare{t}$ to denote the ring of formal power series in $t$, and $\F\indparen{t}$ will denote the formal ring of Laurent series in $t$.
    \item We use $\F(t)\inangle{x_1,\ldots, x_n}$ to denote the ring of polynomials on noncommuting variables $x_1,\ldots, x_n$ with coefficients from the field $\F(t)$. That is, elements of $\F(t)\inangle{x_1,\ldots, x_n}$ are formal finite $\F(t)$-linear combination of \emph{words} over the alphabet $\set{x_1,\ldots, x_n}$ (with $\epsilon$ denoting the empty word); $x_1 x_2 - x_2 x_1$ is thus a nonzero element of $\F(t)\inangle{x_1,\ldots, x_n}$. Here, each $x_i$ commutes with the coefficient field $\F(t)$.  
    \item We use $\partial_{x_i}(f)$ as a shorthand for $\tfrac{\partial f}{\partial x_i}$. For $\F[x_1,\ldots, x_n]$, this would refer to standard partial derivatives. In the noncommutative setting of $\F\inangle{x_1,\ldots, x_n}$, this will refer to \emph{left-derivatives} (which shall be defined in due course). 
    \item We abuse notation and sometimes use $x$ to denote the tuple $(x_1,\ldots, x_n)$. Notations such as $\deg_x$ or $\Hom^{=d}_x$ will denote the degree or homogeneous component with respect to the variables $(x_1,\ldots, x_n)$. 
\end{itemize}

\section{Derivations and associated ``evaluation'' maps}

The core of the PITs involve the following notion of \emph{derivation operators}, and we define a special case that would be sufficient for this paper. 

\begin{definition}[Linear derivations]
    \label{defn:linear-derivations}
    An operator $D: \F(t)[x_1,\ldots, x_n] \rightarrow \F(t)[x_1,\ldots, x_n]$ of the form 
    \begin{align*}
    D f & = \partial_t f + \sum_{i=1}^n b_i(t) \cdot \partial_{x_i}(f)
    \end{align*}
    where $b_i(t) \in \F(t)$ shall be called a \emph{linear derivation} on $\F(t)[x_1,\ldots, x_n]$. We shall call the set
    \[
    \ker D = \setdef{h \in \F(t)[x_1,\dots, x_n]}{Dh = 0}
    \] as the \emph{ring of constants} of the derivation $D$. We shall say that $D$ has a \emph{trivial ring of constants} if $\ker D = \F$. 
\end{definition}

Associated with a derivation of the form $D = \partial_t  + \sum_{i=1}^n b_i(t) \partial_{x_i}$, we can define an $\F(t)$-linear map $\Phi: \F(t)[x_1,\ldots, x_n] \rightarrow \F\indparen{t}$ that satisfies
\begin{equation}\label{eqn:connecting-Phi-to-D}
\frac{d}{dt}\Phi(f) = \Phi(Df).
\end{equation}
Indeed, if $\Phi(f) := f(t, P_1(t), \ldots, P_n(t))$ for some $P_i(t) \in \F\indsquare{t}$, then by the above we must have 
\begin{align*}
\frac{d}{dt}\Phi(f) & = \frac{d}{dt} f(t, P_1(t), \ldots, P_n(t)) = \inparen{\partial_t f + \sum_i P_i'(t) \partial_{x_i} f}(t, P_1(t), \ldots, P_n(t))\\
& = \Phi\inparen{\partial_t f + \sum_i P_i'(t) \partial_{x_i} f} = \Phi(Df)
\end{align*}
which forces $P_i'(t) = b_i(t)$. 
Thus, associated with the derivation $D$ is the ``evaluation map'' that maps each $x_i$ to $P_i(t)$ which is the formal integral of $b_i(t)$. \\

\noindent
The following lemma shows that if $D$ has a trivial ring of constants, then the map $\Phi$ is injective. 

\begin{lemma}[Injectivity of $\Phi$]
    \label{lem:phi-injective}
    Suppose the linear derivation $D$ has a trivial ring of constants. Then, the associated map $\Phi \colon \F(t)[x_1,\ldots, x_n] \rightarrow \F\indparen{t}$ defined above is injective. 
\end{lemma}
\begin{proof}
    Suppose not. Consider the nonzero $f \in \F(t)[x_1,\ldots, x_n]$ with $\Phi(f) = 0$ that is minimal in the following sense:
    \begin{itemize}\itemsep0pt
    \item it has the smallest $\vecx$-degree (say $d$),
    \item it has the fewest monomial of highest $\vecx$-degree.
    \end{itemize}
    Without loss of generality (by scaling $f$ by an element of $\F(t)$ if necessary), we may assume that one of the monomials of $f$ of the highest $x$-degree has coefficient $1$. Let $d = \deg_{\vecx} (f)$.

    Since $\Phi(f) = 0$ we also have $\frac{d}{dt} \Phi(f) = \Phi(Df) = 0$. We claim that $Df$ contradicts minimiality. This is because $Df = \partial_t f + \sum_i b_i(t) \partial_{x_i} f$, and $\sum_i b_i(t) \partial_{x_i} f$ has lower $\vecx$-degree, and $\partial_t f$ has the same $x$-degree but reduces the number of monomials of $\vecx$-degree $d$. Thus, by the minimality of the choice of $f$, we have $Df = 0$. However $D$ has a trivial ring of constants which forces $f \in \F$ and this contradicts either $f \neq 0$ or $\Phi(f) = 0$. 
\end{proof}

Therefore, if we want to find maps $\Phi$ that preserves nonzeroness, we are indirectly looking for suitable derivations $D$ that have a trivial ring of constants. 

\subsection{A concrete derivation with a trivial ring of constants}

There appears to be a rich literature on the study of \emph{derivation fields}. with numerous examples of derivations that have a trivial ring of constants. The derivation we define below appears to be a special case of a broader class of derivations called Shamsuddin derivations. Nowicki's book \cite{nowicki} provides a fairly comprehensive study of this field. \\

\noindent
For this paper, we will focus on the following derivation from here on. 
\begin{align*}
D & = \partial_t + \sum_{i=1}^n b_i(t) \partial_{x_i} \quad \text{ where } b_i(t) = \frac{i}{1 + it}. 
\end{align*}

\begin{lemma}[$D$ has a trivial ring of constants]
    \label{lem:D-trivial-ring-of-constants}
    The derivation $D$ defined above has a trivial ring of constants. That is, if $h(t, x_1,\ldots, x_n) \in R = \F(t)[x_1,\ldots, x_n]$ such that $Dh = 0$, then $h \in \F$. 
\end{lemma}
\begin{proof}
    Let $h_i$ denote the degree $i$ (in $x$) homogeneous  component of $h$ and let $\deg_x h = d$. If $d = 0$, we are done. Thus, we may assume that $d \geq 1$. Focusing on the degree $d$ part of $D(h) = 0$, we have
    \begin{align*}
    0 = \Hom_x^{=d}(D(h)) & = \partial_t h_d\\
    \implies \partial_t h_d & = 0.
    \end{align*}
    Thus, the highest-degree homogeneous part of $h$ is $t$-free. Now focusing on the homogeneous part of degree $d-1$, we have
    \begin{align*}
    0 = \Hom_x^{=d-1}(D(h)) & = \inparen{\partial_t h_{d-1} + \sum_{i=1}^n \frac{i}{1 + it} \cdot \partial_{x_i} h_d }\\
    \implies \partial_t h_{d-1} & = - \sum_{i=1}^n \frac{i}{1+it} \cdot \partial_{x_i} h_d.
    \end{align*}
    Note that $h_d$ is independent of $t$ and thus so is $\partial_{x_i} h_d$. Hence, if $\partial_{x_i} h_d \neq 0$, then the RHS has a simple pole (of order $1$) at $t = -\tfrac{1}{i}$ but $\partial_t h_{d-1}$ can never have simple poles. To see this, write the \emph{partial fraction} decomposition of $h_{d-1} \in \F(t)[x_1,\ldots, x_n]$ as 
    \[
    h_{d-1}(t) = p(t,\vecx)) + \sum_a \sum_{j \geq 1} \frac{c_{a,j}(\vecx)}{(t-a)^j}
    \]
    for some polynomial $p(t, \vecx)$ and $t$-free $c_{a,j}(\vecx) \in \F[x_1,\ldots, x_n]$ for finitely many $a$. Differentiating term by term, $p(t, \vecx)$ contributes a polynomial and thus no poles, while each term $c_{a,j}(\vecx) \cdot (t-a)^{-j}$ contributes $(-j)\cdot  c_{a,j}(\vecx) \cdot (t-a)^{-j-1}$ and thus a pole of order $j+1$ at $a$. Hence, $\partial_t h_{d-1}$ has no pole of order exactly $1$ at any point.) Thus, the only way the above can be zero is if $\partial_{x_i} h_d = 0$ for all $i \in [n]$, which is absurd when $d \geq 1$. Thus, $d = 0$ and hence $h \in \F$. 
\end{proof}

\begin{remark*}
A sufficient condition for the above proof is to use  $b_1(t), \ldots, b_n(t)$ such that each $b_i$ has a distinct simple pole. The above example of $b_i(t) = \frac{i}{1+ it}$ is one of the simplest example of this kind. 
\end{remark*}

The evaluation map associated with the derivation $D$ defined above is $\Phi: x_i \mapsto P_i(t)$ where $P_i(t) = \int b_i(t) dt = \ln(1 + it)$. Thus, \cref{lem:phi-injective} and \cref{lem:D-trivial-ring-of-constants} show that $\setdef{\ln(1 + it)}{i\in [n]}$ is algebraically independent (in fact even over $\F(t)$). 

\section{Wronskians of derivative closed set of polynomials}

Let $\set{f_1,\ldots, f_r}$ be a basis for $\partial^{=*}(f)$, and let us assume that $\deg f_1 \leq \cdots \leq \deg f_r$ with $f_1 = 1$ and $f_r = f$. Recall that $D = \partial_t + \sum_{i=1}^n b_i(t) \partial_{x_i}$ and the $D$-Wronskian of $f_1,\ldots, f_r$ is defined as 
\begin{align*}
\mathcal{W}_D(f_1,\ldots, f_r) & = \begin{bmatrix}
f_1 & \cdots & f_r\\
D f_1 & \cdots & D f_r\\
\vdots & \ddots & \vdots \\
D^{r-1} f_1 & \cdots & D^{r-1} f_r
\end{bmatrix}\\
W_D(f_1,\ldots, f_r) & = \det \mathcal{W}_D(f_1,\ldots, f_r).
\end{align*}

\begin{lemma}[The $D$-Wronskian is $\vecx$-free]
    \label{lem:D-wronskian-x-free}
    $W_D(f_1,\ldots, f_r)$ defined above, for any derivation of the form $D = \partial_t + \sum b_i(t) \partial_{x_i}$, is independent of $\vecx$. 
\end{lemma}
\begin{proof}
    Since $\set{f_1,\ldots, f_r}$ is a basis for $\partial^{=*}(f)$, for each $i \in [n]$ and $j \in [r]$, note that 
    \begin{align*}
    \partial_{x_i} f_j &= \sum_{k < j} a_{ijk} \cdot f_k\\
    \implies \partial_{x_i} D^m f_j & = \sum_{k < j} a_{ijk} \cdot D^m f_k \quad \text{since $\partial_{x_i} Dg = D(\partial_{x_i} g)$ for any $g$.}
    \end{align*}
    Therefore, by the multilinearity of the determinant across columns, $\partial_{x_i} W_D(f_1,\ldots, f_r)$ is a sum of determinants of the matrix $\mathcal{W}_D(f_1,\ldots, f_r)$ with a column replaced by its partial derivative with respect to $x_i$. Thus, if $C_j$ is the $j$-th column of $\mathcal{W}_D(f_1,\ldots, f_r)$, then 
    \[
    \partial_{x_i} C_j = \sum_{k < j} a_{ijk} \cdot C_k. 
    \]
    Therefore, $\partial_{x_i} W_D(f_1,\ldots, f_r) = 0$ for every $i \in [n]$. Hence, $W_D(f_1,\ldots, f_r)$ is independent of any of the $\vecx$-variables and is a function of $t$ alone. 
\end{proof}

The above lemma, along with properties of the derivation $D$ defined in the previous section is sufficient to obtain a hitting set generator for polynomials with low partial derivative space. However, we can make a stronger observation that would be more useful when dealing with noncommutative ABPs in the later section. 

\begin{theorem}\label{thm:wroksian-factorization}
There is an $r\times r$ matrix $F(x)$ and an $r\times r$ matrix $R(t)$ such that $\mathcal{W}_D=R(t)F(x)$. Moreover, the matrix $F(x)$ is upper triangular matrix with diagonal entries equal to $1$.  
\end{theorem}

We can clearly see that \cref{lem:D-wronskian-x-free} is an immediate consequence of this theorem. 

\begin{proof}
Let $\vecy=\{y_1, y_2, \ldots, y_n\}$ be a set of new variables. Use Taylor series expansion to write (for $1\leq j\leq r$)
\[
f_j(\vecy + \vecx)=\sum_{\vece\in \mathbb{N}^n} \frac{\vecx^{\vece}}{\vece!}\cdot  \partial^{\vece} f_j(y).  
\]
Using the derivative closed property we can further write
\begin{equation}\label{eq:vector-notation-1}
f_j(\vecy + \vecx) = f_j(\vecy) + \sum_{\ell < j} f_\ell(\vecy) \cdot F_{\ell,j}(\vecx), 
\end{equation}
for some polynomials $F_{\ell,j}(x)$. 
It is convenient to express these equations in matrix-vector notation. For that define $\widehat{f}(\vecx) = \begin{bmatrix} f_1(\vecx) & f_2(\vecx) & \cdots & f_r(\vecx)\end{bmatrix}$.
Then from the equation above, it is straightforward to express 
\begin{equation}\label{eq:vector-notation-2}
\widehat{f}(\vecy + \vecx)=\widehat{f}(\vecy)\cdot F(\vecx),
\end{equation}
where $F(\vecx)$ is the upper-triangular whose $(\ell,j)$-th entry is just $F_{\ell,j}(\vecx)$ from \eqref{eq:vector-notation-1} with $F_{\ell,\ell} = 1$. Define $D_{\vecy} = \partial_t + \sum_{i=1}^n b_i(t) \partial_{y_i}$ (the analogue of $D$ but with respect to $\vecy$), applying $D_\vecy$ to \eqref{eq:vector-notation-2} (and using the fact that $F(\vecx)$ is independent of $t$ or $\vecy$), we get
\begin{align*}
D_\vecy\inparen{\widehat{f}(\vecy + \vecx)} &= D_\vecy\inparen{\widehat{f}(\vecy)\cdot F(\vecx)} = D_\vecy\inparen{\widehat{f}(\vecy)} \cdot F(\vecx).
\end{align*}
Iterating the above, 
\begin{align*}
D_{\vecy}^m \inparen{\widehat{f}(\vecy + \vecx)} & = D^m_\vecy\inparen{\widehat{f}(\vecy)} \cdot F(\vecx) \quad \text{for all $m \geq 0$}. 
\end{align*}
In other words, if $W(t, \vecx) = W_D(f_1,\ldots, f_r)$, we have $W(t, \vecx + \vecy) = W(t, \vecy) \cdot F(\vecx)$. 
Setting $\vecy = 0$, we get the factorisation $W(t, \vecx) = W(t, \mathbf{0}) \cdot F(\vecx)$. 
\end{proof}

\begin{corollary}[Degree bounds for $W_D(f_1,\ldots, f_r)$]
    \label{cor:degree-bounds-wronskian}
    For $D = \partial_t + \sum \frac{i}{1 + it}\partial_{x_i}$, and $\set{f_1,\ldots, f_r}$ that is $\partial_{\vecx}$-closed, we have 
    \[
    W_D(f_1,\ldots, f_r) = \frac{A(t)}{B(t)}
    \]
    where $A, B \in \F[t]$ with $\deg A, \deg B \leq nr^2$. 
\end{corollary}
\begin{proof}
Let $Q(t) = \prod_{i=1}^n (1 + it)$. We prove by induction that for any $g(\vecx) \in \F[\vecx]$, we have $D^m g = \frac{A_m(t,x)}{Q^e}$ for some $A_m(t,x) \in \F[t,\vecx]$ with $\deg_t A_m \leq m(n-1)$ and $e = m$. 
\begin{align*}
Dg & = \partial_t g + \sum_{i=1}^n \frac{i}{1 + it} \partial_{x_i} g = \frac{\sum_{i=1}^n (iQ(t)/(1+it)) \cdot \partial_{x_i} g}{Q(t)} =: \frac{A_1(t,\vecx)}{Q(t)^{e_1}}.
\end{align*}
Clearly, $\deg_t A_1 \leq n-1$ and $e_1 = 1$. For induction, 
\begin{align*}
D^m g & = D \inparen{\frac{A_{m-1}}{Q^{m-1}}} = \partial_t\inparen{\frac{A_{m-1}}{Q^{m-1}}} + \frac{1}{Q^m} \cdot \inparen{\sum_{i=1}^n \frac{iQ}{(1+it)} \cdot \partial_{x_i} A_{m-1}}\\
& = \frac{\partial_t A_{m-1}}{Q^{m-1}} - \frac{e A_{m-1} \cdot \partial_t Q}{Q^{m}} + \frac{1}{Q^m} \cdot \inparen{\sum_{i=1}^n \frac{iQ}{(1+it)} \cdot \partial_{x_i} A_{m-1}}\\
& = \frac{Q \cdot \partial_t A_{m-1}}{Q^{m}} - \frac{e A_{m-1} \cdot \partial_t Q}{Q^{m}} + \frac{1}{Q^m} \cdot \inparen{\sum_{i=1}^n \frac{iQ}{(1 + it)} \cdot \partial_{x_i} A_{m-1}}\\
& =: \frac{A_m(t,\vecx)}{Q^m}
\end{align*}
Thus, we get that $\deg_t A_m \leq \deg_t A_{m-1} + (n-1)$ and this completes the inductive step. 

From the above bound on the $t$-degree of every entry of $\mathcal{W}_D(f_1,\ldots, f_r)$, and since \cref{lem:D-wronskian-x-free} asserts its determinant is $\vecx$-free, we get that 
\[
W_D(f_1,\ldots, f_r) = \frac{A(t)}{B(t)}
\]
for some $A,B \in \F[t]$ with $\deg A, \deg B \leq n r^2$. 
\end{proof}

\section{Hitting sets for polynomials with low partial derivative space}

Now we have all the ingredients to prove \cref{thm:main-thm}. We state it here in a slightly different form from which \cref{thm:main-thm} would readily follow. 

\begin{theorem}
    \label{thm:main-thm-modified}
    Suppose $0 \neq f(x_1,\ldots, x_n) \in \F[x_1,\ldots, x_n]$ such that $\dim \partial^{=*}(f) \leq r$. Let $P_1(t),\ldots, P_n(t) \in \F\indsquare{t}$ be defined as
    \[
    P_i(t) = it - \frac{(it)^2}{2} + \frac{(it)^3}{3} - \frac{(it)^4}{4} + \cdots  = \ln(1 + it)
    \]
    Then, $f(P_1(t),\cdots, P_n(t)) \neq 0 \bmod t^{B}$ for $B = nr^2 + r$.
\end{theorem}

\cref{thm:main-thm} follows from the above since we may truncate each $P_i$ to the first $B$ terms and retain the nonzeroness of $f$. 

\begin{proof}[Proof of \cref{thm:main-thm-modified}]
    For simplicity, let us assume that $\dim \partial^{=*} f = r$ (else we will replace $r$ with the actual dimension in the proof below). Let $f_1, f_2, \ldots, f_r$ be a basis of the space $V_f$ of all partial derivatives of $f$ in increasing order of degree with $f_r=f$. Let $P(t) = (P_1(t), P_2(t), \ldots, P_n(t))$ where $P_i(t) = \ln(1 + it)$ for $i = 1, \ldots, n$ and let $\Phi: \F(t)[x_1,\ldots, x_n] \rightarrow \F\indparen{t}$ given by $\Phi(g(t, x_1,\ldots, x_n)) = g(t, P_1,\ldots, P_n)$.

    By \cref{lem:phi-injective} and \cref{lem:D-trivial-ring-of-constants}, we have that $\Phi$ is injective. Therefore, $\Phi(f_1), \ldots, \Phi(f_r)$ are linearly independent elements of $\F\indsquare{t}$. By \cite{BD10}, we have 
    \[
    W(\Phi(f_1), \ldots, \Phi(f_r)) = \det \mathcal{W}(\Phi(f_1), \ldots, \Phi(f_r)) \neq 0. 
    \]
    By since $\mathcal{W}(\Phi(f_1), \ldots, \Phi(f_r))$ is just the matrix $W_D(f_1,\ldots, f_r)$ with $\Phi$ applied to every entry, we get from \cref{cor:degree-bounds-wronskian} that 
    \[
    W(\Phi(f_1), \ldots, \Phi(f_r)) = \frac{A(t)}{B(t)}
    \]
    for some nonzero polynomials $A,B$ of degree at most $nr^2$. 



Now suppose $\Phi(f) = 0 \bmod t^B$ for $B = nr^2 + r$,
we have $\frac{d^m}{dt^m} \Phi(f) = 0 \bmod t^{B - m}$
for all $0 \leq m \leq r-1$.
Thus, the entire last column of
$\mathcal{W}(\Phi(f_1), \ldots, \Phi(f_r))$
is divisible by $t^{B - r + 1}$ and hence we must have that
$t^{B - r + 1}$ divides $A(t)$.
But this is impossible since $\deg A \leq nr^2 < B - r + 1$.

Thus, we must have that
$\Phi(f) = f(P_1(t), \ldots, P_n(t))
\neq 0 \bmod t^{nr^2 + r}$.
    
\end{proof}

\section{Extensions to noncommutative ABPs}

The approach for polynomials with low-partial derivative space extends to noncommutative ABPs as well. The ring we are operating on is $\F(t)\inangle{x_1,\ldots, x_n}$. We will be using the same derivation\footnote{Strictly speaking, the above map $D$ is not a `\emph{derivation}' in the traditional sense since it does not satisfy the Leibniz rule (i.e., $D(fg) \neq f D(g) + D(f) g$) due to the non-commutativity. Nevertheless, we abuse this terminology and refer to it as a `derivation' for the sake of exposition.} $D$ on this ring, defined by
\[
D = \partial_t  + \sum_{i=1}^n b_i(t) \cdot \partial_{x_i}
\]
but here $\partial_{x_i}$ refers to the \emph{left-derivative} with respect to $x_i$. That is, for any word $w$, define 
\[
\partial_{x_i}(x_j w) = \begin{cases}
    w & \text{if $i = j$,}\\
    0 & \text{otherwise}. 
\end{cases}
\]

We can prove the analogue of \cref{lem:D-trivial-ring-of-constants} even in the context of $R = \F(t)\inangle{x_1,\ldots, x_n}$. 

\begin{lemma}[$D$ has a trivial ring of constants, noncommutatively]
    \label{lem:D-trivial-ring-of-constants-noncomm}
    The derivation $D: \F(t)\inangle{x_1,\ldots, x_n} \rightarrow \F(t)\inangle{x_1,\ldots, x_n}$ defined above has a trivial ring of constants. That is, if $h\in R$ such that $Dh = 0$, then $h \in \F$. 
\end{lemma}
\begin{proof}
    As earlier, let $h_i$ denote the degree $i$ (in $x$) homogeneous component of $h$ and say $\deg_\vecx h = d$. As in the proof of \cref{lem:D-trivial-ring-of-constants}, since $D(h) = 0$ we can conclude that $h_d$ is $t$-free. 

    For any word $w$, let us use $c_w$ to denote the coefficient of the word $w$ in $h$ (which is an element of $\F(t)$). If we take any word $w$ of length $d-1$, then
    \begin{align*}
    0 = \operatorname{coeff}_w(Dh) & = \partial_t c_w + \sum_{i=1}^n b_i(t) \cdot c_{x_i w}\\
    \implies - \partial_t c_w & = \sum_{i=1}^n \frac{i}{1 + it} \cdot c_{x_i w}.
    \end{align*}
    Once again, the LHS cannot have any simple poles and hence $c_{x_i w} = 0$ for all $i \in [n]$. Since $w$ is arbitrary, this forces $h_d = 0$ which is absurd. Thus, $d = 0$ and hence $h \in \F$. 
\end{proof}

\subsection{The ``evaluation'' map related to $D$ in the noncommutative setting}

As earlier, associated with the above $D$, we wish to define an $\F(t)$-linear map $\Psi: \F(t)\inangle{x_1,\ldots, x_n} \rightarrow \F\indparen{t}$ such that for any $f \in R = \F(t)\inangle{x_1,\ldots, x_n}$
\[
\frac{d}{dt} \Psi(f) = \Psi(Df).
\]
Unlike the commutative case, this is not merely a substitution $x_i \mapsto P_i(t)$ extended to all words since then $\Psi(x_i x_j - x_jx_i) = 0$ and thus $\Psi$ is not injective on $R$. Therefore, we need to define $\Psi$ on each word $w$. Again, if $\Psi(1) = 1$, then we are forced to define $\Psi$ inductively on larger words as
\begin{align*}
\frac{d}{dt}\Psi(x_i w) & = \Psi(D(x_iw))\\
& = \Psi(b_i(t) w) = b_i(t) \Psi(w)\\
\implies \Psi(x_i w) & = \int b_i(t) \Psi(w) dt. 
\end{align*}
Thus, if $T_i$ is the linear operator on $\F\indsquare{t} \rightarrow \F\indsquare{t}$ defined as 
\[
T_i \colon g(t) \mapsto \int b_i(t) g(t) dt
\]
then the above map $\Psi$ can be thought of as substituting in place of each $x_i$ the operator $T_i$. \\

\noindent
Exactly as in \cref{lem:phi-injective}, we now can show that $\Psi$ defined above is injective. 

\begin{lemma}[$\Psi$ is injective]\label{lem:psi-injective}
    Suppose the linear derivation $D$ has a trivial ring of constants in $R = \F(t)\inangle{x_1,\ldots, x_n}$. Then, the associated map $\Psi \colon R \rightarrow \F\indparen{t}$ defined above is injective. 
\end{lemma}
\begin{proof}
Follows the proof of \cref{lem:phi-injective} verbatim. 
\end{proof}

\begin{remark*}
An important distinction here from the commutative case is that the map $\Phi$ defined there was infact an algebra homomorphism. However, the map here is just an $\F(t)$-linear map. This is of course necessary for injectivity as any algebra homomorphism from $\F(t)\inangle{x_1,\ldots, x_n}$ to $\F\indparen{t}$ must necessarily satisfy $\Psi(x_1x_2) = \Psi(x_2x_1)$ and hence cannot be injective. Nevertheless, the map $\Psi$ constructed about will eventually be thought of as an algebra homomorphism to the space of \emph{linear maps on $\F\indsquare{t}$} (instead of to $\F\indsquare{t}$). 
\end{remark*}

\subsection{Noncommutative ABPs have low-dimensional partial derivative space}

\begin{definition}[Noncommutative Algebraic Branching Programs (ncABPs) \cite{N91}]
A \emph{noncommutative algebraic branching program (ncABP)}  of size $s$ is given by a directed acyclic graph of $s$ vertices $\set{v_1,\ldots, v_s}$ with directed edges $(v_i, v_j)$ with $i < j$ having weight $\ell_{ij}$ that is a linear polynomial in $x_1,\ldots, x_n$. 

The polynomial computed by the ncABP is the sum, over all directed paths from $v_1$ to $v_s$, of the product of the edge weights. Concretely, 
\[
f(x_1,\ldots, x_n) = \sum_{v_{i_0} = v_1, v_{i_1}, \cdots, v_{ir} = v_s} \prod_{j=0}^{t-1} \ell_{i_j,i_{j+1}}
\]
More generally, we can define $f_{ij}(x_1,\ldots, x_n)$ analogous as the polynomial computed between vertex $v_i$ and $v_j$ by a similar sum over all directed paths from $v_i$ to $v_j$. 

Such an ncABP is said to be \emph{homogeneous} of degree $d$ if the underlying directed acyclic graph is layered into $d+1$ layers with $v_1$ being the unique vertex in layer $1$ and $v_s$ being the unique vertex in layer $d+1$, and all edges only going from vertices in layer $i$ to layer $i+1$ with all edge-weights $\ell_{ij}$ being linear forms in $x_1,\ldots, x_n$ (i.e., $\ell_{ij}$ has no constant term).
\end{definition}

It is well-known \cite{N91,RS04} that if $f$ is a degree $d$ polynomial that is computed by an ncABP of size $s$, then each homogeneous component of $f$ can be computed by an ncABP of size at most $O(sd^2)$. Furthermore for the purposes of PIT, it suffices to focus on just homogeneous polynomials computable by a homogeneous ABP. See \cite{RS04,FS13} for details. Thus, from here on, we will assume that $f$ is a homogeneous degree $d$ polynomial that is computable by a homogeneous ncABP of size $s$. \\

\noindent
The following lemma is implicit in the work of Raz and Shpilka~\cite{RS04}. 
\begin{lemma}[ncABPs are derivative closed~\cite{RS04}]
    \label{lem:nc-abp-derivative-closed}
    Suppose $f$ is computable by a noncommutative branching program of size $s$. 
    If $\partial^{=*}_\vecx(f)$ refers to span of all left-derivatives of $f$ of all order, then $\partial^{=*}(f)$ is contained in the span of $\setdef{f_{i,s}}{i\in [s]}$. Furthermore, if the ncABP is minimal (in the sense that there is no smaller ABP for $f$), then $\setdef{f_{i,s}}{i \in [s]}$ is a linearly independent set. Similarly, if the ncABP is minimal and homogeneous, then $f_{1,1},\ldots, f_{1,s}$ are linearly independent. 
\end{lemma}

\subsection{Factorisation of the Wronskian matrix}

Let $f\in \F\langle X\rangle$ have degree $d$, and consider an ABP for $f$. For each $i \in [n]$, let $M_i$ be the \emph{adjacency matrix} of the ABP with respect to variable $x_i$. Concretely, $(M_i)_{j,k}$ is the coefficient of the variable $x_i$ in the edge weight $\ell_{j,k}$. Note that $M_i$ is a strictly upper-triangular matrix. Define $F$ as the $s\times s$ matrix where $F_{ij} = f_{i,j}(x_1,\ldots, x_n)$, the polynomial computed between vertex $v_i$ and $v_j$. In particular $F$ is an upper-triangular matrix with all diagonal entries being $1$ by definition. Let $M = \sum_{i} M_ix_i$ be the adjacency matrix of the ABP.

\begin{observation} 
    \label{obs:nc-F-sum-over-words}
    $F = I + M + M^2 + \cdots + M^{d}$
\end{observation}
\begin{proof}
    This follows from the definition of $F$, combined with the fact that $M$ is the adjacency matrix of the (layered) ABP.
\end{proof}

\begin{observation}
    \label{obs:nc-partial_i_f_equals_X_i_F}
If $\partial_{x_i} F$ refers to the matrix obtained by applying the left derivative $\partial_{x_i}$ to every entry of $F$, then 
\[
\partial_{x_i} F = M_i F. 
\]
\end{observation}
\begin{proof}
From \cref{obs:nc-F-sum-over-words}, we have $$F = I + (M_1x_1+ \cdots + M_n x_n) + \cdots + (M_1x_1+ \cdots + M_n x_n)^d$$

Thus, $\partial_{x_i} F = M_i + M_i(M_1x_1+ \cdots + M_n x_n) + \cdots +M_i(M_1x_1 + \cdots + M_nx_n)^{d-1}$. Note that $M_i(M_1x_1+\cdots+M_nx_n)^d=0$: every coefficient in the expansion is $0$ as these are coefficients of the monomials of degree $d+1$ with leftmost variable $x_i$, but no such monomial exists in the ABP. Therefore, $\partial_{x_i} F = M_i F$ as claimed.
\end{proof}

Let us use $f_i$ to denote $f_{1,i}$ defined above. Without loss of generality (by taking the smallest ncABP for $f$), we may assume that $\set{f_1,\ldots, f_s}$ are linearly independent. 

Recall the derivation $D = \partial_t + \sum_{i=1}^n b_i(t) \partial_{x_i}$ and the $D$-Wronskian matrix
\[
\mathcal{W}_D(f_1,\ldots, f_s) = \begin{bmatrix}
f_1 & f_2 & \cdots & f_s\\
D f_1 & D f_2 & \cdots & D f_s\\
\vdots & \vdots & \ddots & \vdots \\
D^{s-1} f_1 & D^{s-1} f_2 & \cdots & D^{s-1} f_s\\
\end{bmatrix}. 
\]

\begin{lemma}[Factorisation of $\mathcal{W}_D(f_1,\ldots, f_s)$]
    \label{lem:nc-W_d-factorisation}
There is an $s\times s$ matrix $R$ consisting of rational functions entries in $t$ such that $\mathcal{W}_D(f_1,\ldots, f_s) = R \cdot F$. 
\end{lemma}
\begin{proof}
    We will construct $R$ one row at a time and show that for all $0 \leq m < s$
    \[
    \begin{bmatrix}D^m f_1 & \cdots & D^m f_s\end{bmatrix} = R_m F
    \]
    for suitably chosen row vector $R_m$. We would then have our factorisation $\mathcal{W}_D(f_1,\ldots, f_s) = R \cdot F$ as required. 

    For $m = 0$, we have $R_0 = [1\; 0\; \cdots \; 0]$. Suppose for some $m \geq 0$ we have
    \begin{align*}
        \begin{bmatrix}D^m f_1 & \cdots & D^m f_s\end{bmatrix} & = R_m F\\
        \implies \begin{bmatrix}D^{m+1} f_1 & \cdots & D^{m+1} f_s\end{bmatrix} & = D\inparen{\begin{bmatrix}D^{m} f_1 & \cdots & D^{m} f_s\end{bmatrix}}\\
    & = D( R_m\cdot F) = \partial_t (R_m \cdot F) + \sum_i b_i(t) \cdot \partial_{x_i} (R_m \cdot F)\\
    & = (\partial_t R_m + R_m B) F \quad \text{where } B = \sum_{i=1}^n b_i(t) M_i\\
    & =: R_{m+1} \cdot F
    \end{align*}
    Thus, $R_{m+1} = \partial_t R_m + R_m B$, and stacking the $R_i$'s into a matrix $R$ gives us the claim. 
\end{proof}

Unlike the commutative case, defining the notion of a determinant of the above matrix needs care as the entries involve noncommuting variables. But the above factorisation is enough to show that the ordinary Wronskian $W(\Psi(f_1), \ldots, \Psi(f_s))$ is a rational function with low-degree numerator and denominator. 

\begin{corollary}\label{cor:det-extension}
Let $\mathcal{W}(t)\coloneqq\mathcal{W}(\Psi(f_1), \ldots, \Psi(f_s))$ be the ordinary Wronskian of $\Psi(f_1), \ldots, \Psi(f_s)$. Then $\mathcal{W}(t)=R(t)\Psi(F)$, where $\Psi(F)$ is the matrix obtained from $F$ by applying the map $\Psi$ entry-wise. As a result $\det(\mathcal{W}(t))=\det(R(t))$. 
\end{corollary}
\begin{proof}
By \cref{lem:nc-W_d-factorisation},  
$
\mathcal{W}_D(f_1, \ldots, f_s)_{a,b}=\sum_{c=1}^sR(t)_{a,c} F_{c,b}. 
$
Since $\Psi$ is $\mathbb{F}(t)$-linear, applying it entry-wise gives 
\[
\Psi(\mathcal{W}_D(f_1, \ldots, f_s)_{a,b})=\sum_{c=1}^sR(t)_{a,c} \Psi(F_{c,b}). 
\]
This also implies that 
\begin{equation}\label{eq:Psi-product}
\Psi(\mathcal{W}_D(f_1, \ldots, f_s))=R(t)\Psi(F).
\end{equation} 
However for $1\leq i\leq s$, 
$
\frac{d}{dt} \Psi(f_i) = \Psi(Df_i).
$
Differentiating $m$ times, we get that 
\[
\frac{d^m}{dt^m} \Psi(f_i) = \Psi(D^mf_i).
\]
Combining this with \cref{eq:Psi-product}, we get that 
$\mathcal{W}(t)=R(t)\Psi(F)$. 
Since $F$ is an upper triangular matrix with diagonal entries $1$, so is $\Psi(F)$, and it follows that 
$\det(\mathcal{W}(t))=\det(R(t))$. 
\end{proof}

\begin{corollary}[Degree bounds for the entries of $R$]
    \label{cor:degree-bounds-R-nc}
If $D = \partial_t + \sum_{i=1}^n \frac{i}{1 + it} \cdot \partial_{x_i}$, then $\det R(t)$ is a rational function of the form $\frac{P(t)}{Q(t)}$ with $\deg P, Q \leq ns^2$. 
\end{corollary}
\begin{proof}
Just as in \cref{cor:degree-bounds-wronskian}, using the recurrence for the rows of $R$ that $R_{i+1} = R_i' + R_i B$, each entry can be shown to be a rational function with numerator and denominator of degree at most $ns$. Since $R$ is an $s\times s$ matrix, we get the required degree bound.
\end{proof}

\subsection{Obtaining the hitting set generator}

\begin{lemma}[Bounding the order of $\Psi(f)$]\label{lem:order-of-Psi}
    Let $f(x_1,\ldots, x_n) \in \F\inangle{x_1,\ldots, x_n}$ be a nonzero polynomial that is computable by an ncABP of size at most $s$. Let $\Psi\colon \F(t)\inangle{x_1,\ldots, x_n} \rightarrow \F\indparen{t}$ as defined above. Then, $\Psi(f) \neq 0 \bmod t^{B}$ for $B = ns^2 + s$. 
\end{lemma}
\begin{proof}
    Assume on the contrary that $\Psi(f) = 0 \bmod t^B$ for $B = ns^2 + s$. Then, for each $m < s$ we have $\frac{d^m}{dt^m} \Psi(f) = 0 \bmod t^{ns^2+1}$. Let $\operatorname{span}\set{f_{1,1},\ldots, f_{1,s}}$ have as a basis $\set{f_1,\ldots, f_s}$ (assuming without loss of generality that $s$ is the minimal ncABP size for $f$) with $f_s = f$. 

    By \cref{cor:det-extension} and \cref{cor:degree-bounds-R-nc}, we have that 
    \[
    W(\Psi(f_1), \ldots, \Psi(f_s)) = \det \mathcal{W}(\Psi(f_1), \ldots, \Psi(f_s)) = \frac{P(t)}{Q(t)}
    \]
    with $\deg P, Q \leq ns^2$. Furthermore, since $\Psi$ is an injective map (\cref{lem:psi-injective}), we know that the above Wronskian determinant is nonzero \cite{BD10}. On the other hand, the last column of $\mathcal{W}(\Psi(f_1), \ldots, \Psi(f_s))$ is divisible by $t^{ns^2+1}$ which is impossible since $\deg P(t) \leq ns^2$. 
\end{proof}

The final step is to convert the above statement into suitable matrices that create the hitting set.

\begin{theorem}
    There is an explicit set of matrices $\tilde{T}_1, \ldots, \tilde{T}_n \in \F^{B \times B}$, for $B = ns^2 + s$ such that any nonzero noncommutative polynomial $f(x_1,\ldots, x_n) \in \F\inangle{x_1,\ldots, x_n}$ that is computable by a homogeneous ncABP of size at most $s$ satisfies $f(\tilde{T}_1,\ldots, \tilde{T}_n) \neq 0$. 
\end{theorem}
\begin{proof}
    Consider the infinite linear operators $T_1,\ldots, T_n \colon \F\indsquare{t} \rightarrow \indsquare{t}$ given by 
    \begin{align*}
    T_i: g(t)  = g_0 + g_1 t + \cdots &\mapsto \int b_i(t) g(t) dt\\
    & = \int i(1 - it + (it)^2 + \cdots) g(t) dt\\
    & = h_0 t + h_1 \frac{t^2}{2} + h_2 \frac{t^3}{3}  + \cdots \\
    \text{where } h(t) & =  h_0 + h_1 t + \cdots  = i (1 - it + (it)^2 - (it)^3) \cdot g(t).
    \end{align*}
    Each $\tilde{T}_i$ is just a truncation of $T_i$ to $\F[t]/t^B$ and thus is an $B \times B$ matrix. 
    
    Recall that $\Psi:\F(t)\inangle{x_1, \ldots, x_n}\rightarrow \F\indparen{t}$ maps words over $\set{x_1,\ldots, x_n}$ to elements of $\F\indsquare{t}$. For length $1$ words, $x_i$, we have $\Psi(x_i) = \int b_i(t) dt = T_i (1)$. More generally for any word $x_{i_1} \cdots x_{i_r}$ we inductively defined 
    \begin{align*}
    \Psi(x_{i_1} \cdots x_{i_r}) & = \int b_{i_1}(t) \Psi(x_{i_2} \cdots x_{i_r}) dt\\
    & = T_{i_1} \inparen{T_{i_2} \cdots T_{i_r}(1)} = T_{i_1} \cdots T_{i_r} (1).
    \end{align*}
    Thus, $\Psi$ maps each word $x_{i_1} \cdots x_{i_r}$ to the power series in $t$ that is the image of the composition $T_{i_1} \cdots T_{i_r}$. Hence, $\Psi(f) = f(T_1,\ldots, T_n)(1)$. Since from \cref{lem:order-of-Psi} we have $\Psi(f) \neq 0 \bmod t^B$, we get that $f(\tilde{T}_1,\cdots, \tilde{T}_n)(1) \neq 0$ and hence $f(\tilde{T}_1,\cdots, \tilde{T}_n)$ is a nonzero linear map and thus a nonzero $B\times B$ matrix. 
\end{proof}

\paragraph{Acknowledgements:} We'd like to thank several members in the community (especially Robert Andrews, Umang Bhaskar, Prerona Chatterjee, Prahladh Harsha, Mrinal Kumar, Amir Shpilka and Anamay Tengse) for numerous philophical discussions on whether or not to publish this. \textcolor{white}{RS and SB also thank the Simons Institute for the Theory of Computing for providing access to ChatGPT Astra to a friend of ours.}

{\let\thefootnote\relax
\footnotetext{\textcolor{\gitinfonotecolour}{\gitinfonote \easteregg}
}}
\bibliographystyle{customurlbst/alphaurlpp}
\bibliography{references,crossref}

@article{Chen1954,
  author  = {Kuo-Tsai Chen},
  title   = {Iterated Integrals and Exponential Homomorphisms},
  journal = {Proceedings of the London Mathematical Society},
  series  = {3},
  volume  = {4},
  pages   = {502--512},
  year    = {1954},
  doi     = {10.1112/plms/s3-4.1.502}
}

@article{Fliess1981,
  author  = {Michel Fliess},
  title   = {Fonctionnelles causales non lin{\'e}aires et
             ind{\'e}termin{\'e}es non commutatives},
  journal = {Bulletin de la Soci{\'e}t{\'e}
             Math{\'e}matique de France},
  volume  = {109},
  pages   = {3--40},
  year    = {1981},
  doi     = {10.24033/bsmf.1931}
}

@article{Moura03,
  author  = {Claire Moura},
  title   = {Bounds on the vanishing order for the solutions of a linear
             differential system},
  journal = {Journal of Dynamical and Control Systems},
  volume  = {9},
  number  = {1},
  pages   = {73--88},
  year    = {2003}
}

@phdthesis{Panzer2015,
  author = {Erik Panzer},
  title  = {Feynman Integrals and Hyperlogarithms},
  school = {Humboldt-Universit{\"a}t zu Berlin},
  year   = {2015},
  doi    = {10.18452/17157},
  note   = {See Chapter 3, Lemma 3.3.5}
}

@inproceedings{DDMS11,
  author    = {Deneufch{\^a}tel, Matthieu and Duchamp, G{\'e}rard H. E. and
               {Hoang Ngoc Minh}, Vincel and Solomon, Allan I.},
  title     = {Independence of Hyperlogarithms over Function Fields via
               Algebraic Combinatorics},
  booktitle = {Algebraic Informatics: 4th International Conference, {CAI} 2011,
               Linz, Austria, June 21--24, 2011, Proceedings},
  editor    = {Winkler, Franz},
  series    = {Lecture Notes in Computer Science},
  volume    = {6742},
  pages     = {127--139},
  publisher = {Springer},
  address   = {Berlin, Heidelberg},
  year      = {2011},
  doi       = {10.1007/978-3-642-21493-6_8},
  eprint    = {1101.4497},
  archivePrefix = {arXiv},
  primaryClass  = {math.CO}
}

@book{BerstelReutenauer2010,
  author    = {Jean Berstel and Christophe Reutenauer},
  title     = {Noncommutative Rational Series with
               Applications},
  publisher = {Cambridge University Press},
  series    = {Encyclopedia of Mathematics and Its
               Applications},
  volume    = {137},
  year      = {2011}
}

@article{GGOW20,
  author  = {Ankit Garg and Leonid Gurvits and
             Rafael Oliveira and Avi Wigderson},
  title   = {Operator Scaling: Theory and Applications},
  journal = {Foundations of Computational Mathematics},
  volume  = {20},
  number  = {2},
  pages   = {223--290},
  year    = {2020},
  doi     = {10.1007/s10208-019-09417-z}
}

@article{HW15,
  author  = {Pavel Hrube{\v{s}} and Avi Wigderson},
  title   = {Non-commutative Arithmetic Circuits with Division},
  journal = {Theory of Computing},
  volume  = {11},
  number  = {14},
  pages   = {357--393},
  year    = {2015},
  doi     = {10.4086/toc.2015.v011a014}
}

@inproceedings{GGOW16,
  author    = {Ankit Garg and Leonid Gurvits and
               Rafael Mendes de Oliveira and Avi Wigderson},
  title     = {A Deterministic Polynomial Time Algorithm for
               Non-commutative Rational Identity Testing},
  booktitle = {FOCS},
  pages     = {109--117},
  year      = {2016},
  doi       = {10.1109/FOCS.2016.95}
}

@article{IQS18,
  author  = {G{\'a}bor Ivanyos and Youming Qiao and
             K. V. Subrahmanyam},
  title   = {Constructive Non-commutative Rank Computation
             Is in Deterministic Polynomial Time},
  journal = {Computational Complexity},
  volume  = {27},
  number  = {4},
  pages   = {561--593},
  year    = {2018},
  doi     = {10.1007/s00037-018-0165-7}
}

@article{HH21,
  author  = {Masaki Hamada and Hiroshi Hirai},
  title   = {Computing the nc-Rank via Discrete Convex
             Optimization on {CAT}(0) Spaces},
  journal = {SIAM Journal on Applied Algebra and Geometry},
  volume  = {5},
  number  = {3},
  pages   = {455--478},
  year    = {2021},
  doi     = {10.1137/20M138836X}
}

@article{DM17,
  author  = {Harm Derksen and Visu Makam},
  title   = {Polynomial Degree Bounds for Matrix Semi-invariants},
  journal = {Advances in Mathematics},
  volume  = {310},
  pages   = {44--63},
  year    = {2017},
  doi     = {10.1016/j.aim.2017.01.018}
}

@inproceedings{ACM24,
  author    = {V. Arvind and Abhranil Chatterjee and
               Partha Mukhopadhyay},
  title     = {Black-Box Identity Testing of Noncommutative
               Rational Formulas in Deterministic
               Quasipolynomial Time},
  booktitle = {STOC},
  pages     = {106--117},
  year      = {2024},
  doi       = {10.1145/3618260.3649693}
}

@proceedings{conf/fsttcs/2005,
  editor    = {Ramaswamy Ramanujam and
               Sandeep Sen},
  title     = {{FSTTCS} 2005: Foundations of Software Technology and Theoretical
               Computer Science, 25th International Conference, Hyderabad, India,
               December 15-18, 2005, Proceedings},
  series    = {Lecture Notes in Computer Science},
  volume    = {3821},
  publisher = {Springer},
  year      = {2005},
  url       = {https://doi.org/10.1007/11590156},
  doi       = {10.1007/11590156},
  isbn      = {3-540-30495-9},
  bibsource = {dblp computer science bibliography, https://dblp.org}
}

@inproceedings{conf/stoc/Nisan91,
  author    = {Noam Nisan},
  title     = {Lower bounds for non-commutative computation (Extended Abstract)},
  booktitle = {Proceedings of the Twenty-Third Annual ACM Symposium on Theory of Computing (STOC '91)},
  pages     = {410--418},
  year      = {1991},
  publisher = {ACM},
  address   = {New Orleans, Louisiana, USA},
  doi       = {10.1145/103418.103462}
}

@proceedings{conf/stoc/1980,
  editor    = {Raymond E. Miller and
               Seymour Ginsburg and
               Walter A. Burkhard and
               Richard J. Lipton},
  title     = {Proceedings of the 12th Annual {ACM} Symposium on Theory of Computing,
               April 28-30, 1980, Los Angeles, California, {USA}},
  publisher = {{ACM}},
  year      = {1980},
  bibsource = {dblp computer science bibliography, https://dblp.org}
}

@STRING{ccc	= "Conference on Computational Complexity (CCC)" }

@STRING{eccc	= "Electronic Colloquium on Computational Complexity (ECCC)" }

@STRING{focs	= "Foundations of Computer Science (FOCS)" }

@STRING{fsttcs	= "Foundations of Software Technology and Theoretical
		  Computer Science Science (FSTTCS)" }

@STRING{siamjc	= "SIAM Journal of Computing" }

@STRING{stoc	= "Symposium on Theory of Computing (STOC)" }

@article{BD10,
  author       = {Alin Bostan and
                  Philippe Dumas},
  title        = {Wronskians and Linear Independence},
  journal      = {Am. Math. Mon.},
  volume       = {117},
  number       = {8},
  pages        = {722--727},
  year         = {2010},
  url          = {https://doi.org/10.4169/000298910X515785},
  doi          = {10.4169/000298910X515785},
}

@Misc{astra_proof,
  author  = "OpenAI GPT 6 Astra",
  title = "Building a polynomial hitting set",
  year    = 2026,
  url = {https://www.tcs.tifr.res.in/~ramprasad/assets/pubs/astra-low-pds-hsg.pdf}
}

@article{Moura_2004,
   title={On the multiplicity of hyperelliptic integrals},
   volume={17},
   ISSN={1361-6544},
   url={http://dx.doi.org/10.1088/0951-7715/17/6/004},
   DOI={10.1088/0951-7715/17/6/004},
   number={6},
   journal={Nonlinearity},
   publisher={IOP Publishing},
   author={Moura, Claire},
   year={2004},
   month=Aug, pages={2057-2068} }

@inproceedings{GOSSS26,
  author       = {Abhibhav Garg and
                  Rafael Oliveira and
                  Akash Kumar Sengupta and
                  Nir Shalmon and
                  Amir Shpilka},
  editor       = {Dana Moshkovitz},
  title        = {Rank Bounds and Polynomial-Time {PIT} for {\(\Sigma\)}{\^{}}k {\(\Pi\)}
                  {\(\Sigma\)} {\(\Pi\)}{\({^2}\)} Circuits},
  booktitle    = {41st Computational Complexity Conference, {CCC} 2026, Lisbon, Portugal,
                  August 3-6, 2026},
  series       = {LIPIcs},
  volume       = {383},
  pages        = {17:1--17:18},
  publisher    = {Schloss Dagstuhl - Leibniz-Zentrum f{\"{u}}r Informatik},
  year         = {2026},
  url          = {https://doi.org/10.4230/LIPIcs.CCC.2026.17},
  doi          = {10.4230/LIPICS.CCC.2026.17},
  bibsource    = {dblp computer science bibliography, https://dblp.org}
}

@inproceedings{GG20,
  author       = {Zeyu Guo and
                  Rohit Gurjar},
  editor       = {Jaroslaw Byrka and
                  Raghu Meka},
  title        = {Improved Explicit Hitting-Sets for ROABPs},
  booktitle    = {Approximation, Randomization, and Combinatorial Optimization. Algorithms
                  and Techniques, {APPROX/RANDOM} 2020, Virtual Conference, August 17-19,
                  2020},
  series       = {LIPIcs},
  volume       = {176},
  pages        = {4:1--4:16},
  publisher    = {Schloss Dagstuhl - Leibniz-Zentrum f{\"{u}}r Informatik},
  year         = {2020},
  url          = {https://doi.org/10.4230/LIPIcs.APPROX/RANDOM.2020.4},
  doi          = {10.4230/LIPICS.APPROX/RANDOM.2020.4},
  bibsource    = {dblp computer science bibliography, https://dblp.org}
}

@book{nowicki,
author = {Nowicki, Andrzej},
year = {1994},
month = {10},
pages = {},
title = {Polynomial derivations and their rings of constants},
isbn = {83-231-0543-X},
doi = {10.13140/2.1.3906.3844}
}

@inproceedings{BKRRSS26,
  author       = {Somnath Bhattacharjee and
                  Mrinal Kumar and
                  Shanthanu S. Rai and
                  Varun Ramanathan and
                  Ramprasad Saptharishi and
                  Shubhangi Saraf},
  editor       = {Aditya Bhaskara and
                  Artur Czumaj},
  title        = {Closure under Factorization from a Result of Furstenberg},
  booktitle    = {Proceedings of the 58th Annual {ACM} Symposium on Theory of Computing,
                  {STOC} 2026, Salt Lake City, UT, USA, June 22-26, 2026},
  pages        = {174--185},
  publisher    = {{ACM}},
  year         = {2026},
  url          = {https://doi.org/10.1145/3798129.3800738},
  doi          = {10.1145/3798129.3800738},
  bibsource    = {dblp computer science bibliography, https://dblp.org}
}

@article{Oliveira16,
  author       = {Rafael Oliveira},
  title        = {Factors of low individual degree polynomials},
  journal      = {Comput. Complex.},
  volume       = {25},
  number       = {2},
  pages        = {507--561},
  year         = {2016},
  url          = {https://doi.org/10.1007/s00037-016-0130-2},
  doi          = {10.1007/S00037-016-0130-2}
}

@inproceedings{AF22,
  author       = {Robert Andrews and
                  Michael A. Forbes},

  title        = {Ideals, determinants, and straightening: proving and using lower bounds
                  for polynomial ideals},
  booktitle    = {{STOC} '22: 54th Annual {ACM} {SIGACT} Symposium on Theory of Computing,
                  Rome, Italy, June 20 - 24, 2022},
  pages        = {389--402},
  publisher    = {{ACM}},
  year         = {2022},
  url          = {https://doi.org/10.1145/3519935.3520025},
  doi          = {10.1145/3519935.3520025},
}

@inproceedings{GKSS19,
  author       = {Zeyu Guo and
                  Mrinal Kumar and
                  Ramprasad Saptharishi and
                  Noam Solomon},
  editor       = {David Zuckerman},
  title        = {Derandomization from Algebraic Hardness: Treading the Borders},
  booktitle    = {60th {IEEE} Annual Symposium on Foundations of Computer Science, {FOCS}
                  2019, Baltimore, Maryland, USA, November 9-12, 2019},
  pages        = {147--157},
  publisher    = {{IEEE} Computer Society},
  year         = {2019},
  url          = {https://doi.org/10.1109/FOCS.2019.00018},
  doi          = {10.1109/FOCS.2019.00018}
}

@article{ChouKS19,
 author = {Chou, Chi-Ning and Kumar, Mrinal and Solomon, Noam},
 title = {Closure Results for Polynomial Factorization},
 year = {2019},
 pages = {1--34},
 doi = {10.4086/toc.2019.v015a013},
 publisher = {Theory of Computing},
 journal = {Theory of Computing},
 volume = {15},
 number = {13},
 URL = {https://theoryofcomputing.org/articles/v015a013},
}

@InProceedings{	  A05a,
  author	= {Manindra Agrawal},
  title		= {{P}roving {L}ower {B}ounds {V}ia {P}seudo-random
		  {G}enerators},
  booktitle	= {\FSTTCS{2005}},
  crossref      = {conf/fsttcs/2005},
  year		= 2005,
  pages		= {92-105},
  doi		= {10.1007/11590156_6}
}

@Article{	  AGKS15,
  author	= {Manindra Agrawal and Rohit Gurjar and Arpita Korwar and
		  Nitin Saxena},
  title		= {Hitting-Sets for {ROABP} and Sum of Set-Multilinear
		  Circuits},
  journal	= siamjc,
  volume	= 44,
  number	= 3,
  pages		= {669--697},
  year		= 2015,
  url		= {http://dx.doi.org/10.1137/140975103},
  doi		= {10.1137/140975103},
  eprint	= {1406.7535}
}

@inproceedings{ASS13,
  author    = {Manindra Agrawal and
               Chandan Saha and
               Nitin Saxena},
  title     = {Quasi-polynomial hitting-set for set-depth-{\(\Delta\)} formulas},
  booktitle = {\STOC{2013}},
  pages     = {321--330},
  year      = {2013},
  url       = {http://doi.acm.org/10.1145/2488608.2488649},
  doi       = {10.1145/2488608.2488649},
  note      = {\shortECCC{12}{113}}
}

@InProceedings{	  CK97,
  author	= {Zhi-Zhong Chen and Ming-Yang Kao},
  title		= {{Reducing Randomness via Irrational Numbers}},
  booktitle	= {\STOC{1997}},
  year		= 1997,
  pages		= {200-209}
}

@Article{	  DSY09,
  author	= {Zeev Dvir and Amir Shpilka and Amir Yehudayoff},
  title		= {Hardness-Randomness Tradeoffs for Bounded Depth Arithmetic Circuits},
  journal	= {{SIAM} J. Comput.},
  volume	= 39,
  number	= 4,
  pages		= {1279--1293},
  year		= 2009,
  url		= {http://dx.doi.org/10.1137/080735850},
  doi		= {10.1137/080735850}
}

@InProceedings{	  FS13,
  author	= {Michael A. Forbes and Amir Shpilka},
  title		= {Quasipolynomial-Time Identity Testing of Non-commutative and Read-Once Oblivious Algebraic Branching Programs},
  booktitle	= {\FOCS{2013}},
  pages		= {243--252},
  year		= 2013,
  url		= {http://dx.doi.org/10.1109/FOCS.2013.34},
  doi		= {10.1109/FOCS.2013.34},
  note		= {\farXiv{1209.2408}},
  toupdate	= {journal}
}

@InProceedings{	  FSS14,
  author	= {Michael A. Forbes and Ramprasad Saptharishi and Amir Shpilka},
  title		= {Hitting sets for multilinear read-once algebraic branching programs, in any order},
  booktitle	= {\STOC{2014}},
  pages		= {867--875},
  year		= 2014,
  url		= {http://doi.acm.org/10.1145/2591796.2591816},
  doi		= {10.1145/2591796.2591816}
}

@Article{			 FSV18,
  author    = {Michael A. Forbes and Amir Shpilka and Ben Lee Volk},
  title     = {Succinct Hitting Sets and Barriers to Proving Lower Bounds for Algebraic Circuits},
  journal   = {Theory of Computing},
  volume    = {14},
  number    = {1},
  pages     = {1--45},
  year      = {2018},
  url       = {https://doi.org/10.4086/toc.2018.v014a018},
  doi       = {10.4086/toc.2018.v014a018},
  bibsource = {dblp computer science bibliography, https://dblp.org}
}

@Article{			GKSS17,
  author    = {Joshua A. Grochow and Mrinal Kumar and Michael E. Saks and Shubhangi Saraf},
  title     = {Towards an algebraic natural proofs barrier via polynomial identity testing},
  journal   = {CoRR},
  volume    = {abs/1701.01717},
  year      = {2017},
  url       = {http://arxiv.org/abs/1701.01717},
  archivePrefix = {arXiv},
  eprint    = {1701.01717},
  bibsource = {dblp computer science bibliography, https://dblp.org}
}

@InProceedings{	  GKST15,
  author	= {Rohit Gurjar and Arpita Korwar and Nitin Saxena and Thomas Thierauf},
  title		= {Deterministic Identity Testing for Sum of Read-once Oblivious Arithmetic Branching Programs},
  booktitle	= {\CCC{2015}},
  pages		= {323--346},
  year		= 2015,
  url		= {http://dx.doi.org/10.4230/LIPIcs.CCC.2015.323},
  doi		= {10.4230/LIPIcs.CCC.2015.323},
  note	        = {\arXiv{1411.7341}},
  toupdate      = {journal}
}

@InProceedings{	  HS80,
  author	= {Joos Heintz and Claus-Peter Schnorr},
  title		= {{Testing Polynomials which Are Easy to Compute (Extended
		  Abstract)}},
  booktitle	= {\STOC{1980}},
  year		= 1980,
  crossref      = {conf/stoc/1980},
  pages		= {262-272},
  doi		= {10.1145/800141.804674}
}

@Article{	  K10,
  author	= {Neeraj Kayal},
  title		= {{Algorithms for Arithmetic Circuits}},
  journal	= eccc,
  year		= 2010,
  ee		= {http://www.eccc.uni-trier.de/report/2010/073/}
}

@Article{	  KI04,
  author	= {Valentine Kabanets and Russell Impagliazzo},
  title		= {{D}erandomizing Polynomial Identity Tests Means Proving
		  Circuit Lower Bounds},
  journal	= {Computational Complexity},
  volume	= 13,
  number	= {1-2},
  year		= 2004,
  pages		= {1-46},
  doi		= {10.1007/s00037-004-0182-6},
  note		= {\pSTOC{2003}}
}

@inproceedings{LST21,
  author       = {Nutan Limaye and
                  Srikanth Srinivasan and
                  S{\'{e}}bastien Tavenas},
  title        = {Superpolynomial Lower Bounds Against Low-Depth Algebraic Circuits},
  booktitle    = {\FOCS{2021}},
  pages        = {804--814},
  publisher    = {{IEEE}},
  year         = {2021},
  url          = {https://doi.org/10.1109/FOCS52979.2021.00083},
  doi          = {10.1109/FOCS52979.2021.00083},
  note         = {\pECCC{21}{081}},
  bibsource    = {dblp computer science bibliography, https://dblp.org}
}

@InProceedings{	  LV98,
  author	= {Daniel Lewin and Salil P. Vadhan},
  title		= {{Checking Polynomial Identities over any Field: Towards a
		  Derandomization?}},
  booktitle	= {\STOC{1998}},
  year		= 1998,
  pages		= {438-447}
}

@InProceedings{	  N91,
  author	= {Noam Nisan},
  title		= {{Lower bounds for non-commutative computation}},
  booktitle	= {\STOC{1991}},
  year		= 1991,
  pages		= {410-418},
  doi		= {10.1145/103418.103462},
  note		= {Available on
		  \href{http://citeseerx.ist.psu.edu/viewdoc/summary?doi=10.1.1.17.5067}{\tt
		  citeseer:10.1.1.17.5067}}
}

@InProceedings{	  RS04,
  author	= {Ran Raz and Amir Shpilka},
  title		= {{D}eterministic {P}olynomial {I}dentity {T}esting in
		  {N}on-{C}ommutative {M}odels},
  booktitle	= {\CCC{2004}},
  year		= 2004,
  pages		= {215-222}
}

@Article{	  RS05,
  author	= {Ran Raz and Amir Shpilka},
  title		= {Deterministic polynomial identity testing in
		  non-commutative models},
  journal	= {Computational Complexity},
  volume	= 14,
  number	= 1,
  pages		= {1--19},
  year		= 2005,
  url		= {http://dx.doi.org/10.1007/s00037-005-0188-8},
  doi		= {10.1007/s00037-005-0188-8},
  note		= {\pCCC{2004}}
}

@InProceedings{	  S08b,
  author	= {Nitin Saxena},
  title		= {{Diagonal Circuit Identity Testing and Lower Bounds}},
  booktitle	= {\ICALP{2008}},
  year		= 2008,
  pages		= {60-71},
  doi		= {10.1007/978-3-540-70575-8_6},
  eccc		= {TR07/124},
  ecccurlid	= {2007/124}
}

@article{Schutzenberger61,
  author  = {Marcel-Paul Sch{\"u}tzenberger},
  title   = {On the Definition of a Family of Automata},
  journal = {Information and Control},
  volume  = {4},
  number  = {2--3},
  pages   = {245--270},
  year    = {1961},
  doi     = {10.1016/S0019-9958(61)80020-X}
}

@misc{OpenAI26,
  author = {{OpenAI}},
  title = {{One Rational Matrix Hitting Point for Noncommutative Formulas}},
  howpublished = {OpenAI Math Release preprint
                  \href{https://github.com/openai/math/blob/main/preprints/One-Rational-Matrix-Hitting-Point-for-Noncommutative-Formulas-September-24-2026/One-Rational-Matrix-Hitting-Point-for-Noncommutative-Formulas-September-24-2026.pdf}{OAI:One-Rational-Matrix-Hitting-Point-for-Noncommutative-Formulas-September-24-2026}},
  year = {2026}
}

@article{espinosa2026chen,
  title={What is a Chen--Fliess Series?},
  author={Espinosa, Luis A Duffaut and Ebrahimi-Fard, Kurusch and Gray, W Steven and Venkatesh, GS},
  journal={Notices of the American Mathematical Society},
  volume={73},
  number={2},
  year={2026},
  url={https://www.ams.org/journals/notices/202602/rnoti-p144.pdf}
}

\appendix

\section{Tracing back the classical methods}\label{sec:classical-proof}

In this section, we provide a version of OpenAI's proof \cite{OpenAI26} of black-box hitting point for noncommutative ABPs in the language of the literature it utilizes.

Let $\F\langle x_1, \ldots, x_n\rangle$ be the free associative algebra. Given a polynomial $f\in\F\langle x_1, \ldots, x_n\rangle$ computed by an ABP of size $s$, it is a standard fact that $f$ has an $s$-dimensional linear representation consisting of $u,v \in \F^s$, and jointly nilpotent matrices $A_i\in \F^{s\times s}$ such that 
\[
f=\sum_w (u^{T} A_w v) w, 
\]
where $w=x_{i_1}\ldots x_{i_k}$ and $A_w=A_{i_1}\ldots A_{i_k}, ~A_{\emptyset}=I_s$. Moreover, $A_w=0$ for $|w|>s$. This follows from the standard linear-representation point of view of recognizable series \cite{BerstelReutenauer2010, Schutzenberger61}. 

For $1\leq i\leq n$, define 
\[
b_i(t)=\frac{i}{1-it}, \qquad (T_i h)(t)=\int_0^t b_i(z)h(z) dz, 
\]
where the integration is coefficient-wise for $h(t)\in \F[[t]]$. For every word define the iterated integrals as follows
\[
I_{\emptyset}=1, \qquad I_{x_iw}(t)=(T_iI_w)(t). 
\]
The linear map $\Psi : \F\langle x_1, \ldots, x_n\rangle \rightarrow \F[[t]]$ is defined as 
\[
\Psi(w)=I_w(t), 
\]
and extending linearly. The map $\Psi$ is well-studied in the context of nonlinear control theory and $\Psi(f)$ is known as Chen-Fliess series \cite{Chen1954, Fliess1981}. The following result is also well-known. 

\begin{lemma}\cite{Panzer2015, DDMS11}\label{lem:independence}
$f\neq 0 \Rightarrow \Psi(f)\neq 0$. 
\end{lemma}

Now construct the matrix-values formal series 
\[
Z(t)=\sum_w I_w(t) A_w \in (\F[[t]])^{s\times s}. 
\]
Differentiating with respect to $t$ and using the easy fact that 
\[
I'_{x_j v}(t)=b_j(t)I_v(t), 
\]
we get the following
\[
Z'(t)=\sum_{j,v} b_j I_v A_j A_v. 
\]
Arranging the terms further
\begin{equation}\label{eqn:diff1}
Z'(t)=\left(\sum_{j=1}^n b_j A_j\right) \left(\sum_v I_v(t)A_v\right)=A(t) Z(t),
\end{equation}
where $A(t)=\sum_{i=1}^n \frac{i A_i}{1-it}$. Let 
$y(t)=Z(t) v$ and 
\[
g(t)=u^T y(t)=\sum_w (u^T A_w v) I_w(t)=\Psi(f)(t).
\]
By \cref{lem:independence}, $g(t)\neq 0$. Moreover using \cref{eqn:diff1}, 
\begin{equation}\label{eqn:diff2}
y'(t)=Z'(t)v=A(t)y(t).
\end{equation}
Note that the rational matrix $A(t)$ has no pole at $t=0$. Given $t_0$ and a Taylor series expansion
\[
h(t)=\sum_{k=0}^{\infty} c_k (t-t_0)^k, 
\]
define 
\[
\ord_{t_0}h=\min\{k\geq 0 : c_k\neq 0\}. 
\]
Our goal is to prove that $\ord_0 g\leq N$ for a polynomially bounded integer $N$. For that we use a classical result by Moura \cite[Theorem 1]{Moura_2004} which is also attributed to \cite{Moura03}.  We state the result in a form suitable for our purpose. Moura states the result over $\mathbb C$, but it easily extends to arbitrary fields of characteristic zero.

\begin{theorem}\cite[Theorem 1]{Moura_2004}\label{thm:moura}
Suppose $y'(t)=A(t)y(t)$ be an $s$-dimensional rational linear system, $t_0$ is not a pole of $A(t)$, and $g(t)=u^T y(t)$ is not identically zero. Let $P$ be the sum of the order of all poles (including infinity) of the differential form $\Omega=A(t)dt$. Then, 
\[
\ord_{t_0} g\leq s-1 + \binom{s}{2}(P-2). 
\]
\end{theorem}

We note a few standard facts. 
\begin{enumerate}
\item Ordinary point: $t_0$ is an ordinary point when each rational entry of $A(t)$ is defined at $t_0$. 

\item Pole at a finite point $a$: the pole order of $A(t) dt$ at $a$ is the largest pole order among the rational functions in $A(t)$.

\item Pole at infinity: substitute $z=1/t$. Then 
$dt=-z^{-2}dz$. The order of the pole at infinity is the order of the pole at $z=0$. 
\end{enumerate}

We apply the above to our case. For us, 
\[
A(t)=\sum_{i=1}^n \frac{iA_i}{1-it}. 
\]
Let $\Omega=A(t)dt$. There are $n$ finite poles of order $1$ for $A(t)$. Substitute $t=\frac{1}{z}$. Then, 
\[
A(t)dt=-\left(\sum_{i=1}^n\frac{iA_i}{z-i}\right)z^{-1}dz. 
\]
So the order of the pole at infinity is at most $1$. Hence $P\leq n+1$. Applying \cref{thm:moura},
\[
\ord_0 g\leq s-1 + \binom{s}{2} (n-1).
\]

Set, \[N\coloneqq s-1 + \binom{s}{2} (n-1).\] 

\paragraph{Truncation: } Each $T_i : \F[[t]]\rightarrow \F[[t]]$ is a linear operator. Define the vector space 
$V_N=\F[[t]]/(t^{N+1})$. Naturally, each $T_i$ induces an $\F$-linear map
\[
T'_i : V_N \rightarrow V_N, \qquad T'_i(h)=T_i(h) \mod t^{N+1}. 
\]
By direct expansion 
\[
T'_i(t^k)=\sum_{r=k+1}^N \frac{i^{r-k}}{r} t^r. 
\]
Hence, relative to the basis $(1, t, \ldots, t^N)$ the matrix $T'_i$ can be represented as 

\begin{equation}\label{eq:explicitB}
 (T'_i)_{r,k}=\begin{cases}
 \dfrac{i^{\,r-k}}{r},& 0\le k<r\le N,\\[5pt]
 0,&\text{otherwise}.
 \end{cases}
 \end{equation} 

Now we are ready to state the main theorem. 

\begin{theorem}\label{thm:appendix-main-1}
For every nonzero noncommutative polynomial $f\in\F\langle x_1, \ldots, x_n\rangle$ computable by an ABP of size at most $s$, 
$
f(T'_1, \ldots, T'_n)\neq 0. 
$
\end{theorem}

\begin{proof}
Using the connection with iterated integrals, the main thing to notice is that 
\[
f(T'_1, \ldots, T'_n) 1 =\sum_{w} c_w I_w(t)=\Psi(f)(t)\mod t^{N+1}. 
\]
However, we already proved that $\ord_0 g(t)\leq N$ where 
\[
g(t)=\sum_w (u^T A_w v) I_w(t). 
\]
Hence $f(T'_1, \ldots, T'_n) 1=g(t) \mod t^{N+1}\neq 0$ and the proof of the theorem follows. 
\end{proof}

\end{document}